\documentclass[journal]{article}                                  

\usepackage{arxiv}
\usepackage{hyperref}
\usepackage{graphicx}          
\usepackage{amsmath} 
\usepackage{amssymb}  

\usepackage{amsthm}
\usepackage{color,soul}
\usepackage{flushend}
\usepackage{cite}
\usepackage{mfirstuc}
\usepackage{subfig}
\usepackage{mathtools}
\usepackage{xcolor}
\usepackage{array,multirow}
\usepackage{algorithm} 
\usepackage{algorithmic}  
\usepackage[linesnumbered,algo2e,ruled,vlined,norelsize]{algorithm2e} 
\allowdisplaybreaks
\newtheorem{thmm}{Theorem}

\newtheorem{assm}{Assumption}

\usepackage{tikz}
\usepackage{textcomp}
\usepackage{hyperref}
\usepackage{lipsum}

\usepackage{authblk}
\title{\LARGE \bf
Detecting Switching Attacks On Traffic Flow Regulation For Changing Driving Patterns
}

\author[1]{Sanchita Ghosh}
\author[1]{Tanushree Roy}
\affil[1]{Department of  Mechanical Engineering, Texas Tech University, Lubbock, TX 79409, US. Emails:~{\tt\small sancghos@ttu.edu, tanushree.roy@ttu.edu}.}

\begin{document}

\maketitle
\thispagestyle{empty}
\pagestyle{empty}


\begin{abstract}
Modern traffic management systems increasingly adopt hierarchical control strategies for improved efficiency and scalability, where a local traffic controller mode is chosen by a supervisory controller based on the changing large-scale driving patterns.  Unfortunately, such local metering controllers are also vulnerable to cyberattacks that can disrupt the controller switching, leading to undesired, inefficient, and even unsafe traffic operations. Additionally, the detection of such attacks becomes challenging when the operational mode of the traffic is uncertain and the operational mode identification is delayed. 
Thus, in this work, we propose a cyberattack detection scheme to detect the compromised controller switching in ramp metering for an uncertain, multimodal macroscopic traffic operation of a freeway segment. In particular, we propose a bank of detectors corresponding to each admissible traffic mode that can compensate for the uncertain traffic mode of the freeway. Furthermore, we utilize backstepping tools along with Lyapunov function theory to achieve analytical performance guarantees for the detector, such as nominal exponential stability, anomaly/uncertainty-to-residual stability, robustness, and sensitivity. Finally, we demonstrate the efficacy of the proposed detection scheme through simulations of free traffic under realistic traffic parameters, uncertainties, and commonly occurring attack scenarios. 
\end{abstract}

\section{Introduction}
Intelligent traffic control integrates advanced real-time sensing, cloud-computation, smart decision making, and wireless communication technologies to ensure smooth traffic operations with extensive automation \cite{cassandras2016smart}. This emerging connectivity also introduces cyberthreats that may lead to large-scale resource disruptions, financial losses, or even potential safety risks \cite{perrine2019implications}. Therefore, detection of these cyberattacks is essential to facilitate effective and reliable traffic control \cite{roy2020secure, roy2021socio}. Yet, the task of cyberattack detection is especially challenging for multi-modal systems with unknown operational modes \cite{ghosh2023security}, such as changing traffic flow. 

The different operational modes for traffic are induced from uncertainties in traffic parameters (e.g., free flow velocity or traffic demand), where each mode can be interpreted as a distinct curve in a family of fundamental diagram curves \cite{fan2013comparative}. Several factors, including driver behavior, road condition, weather, or traffic incidents, contribute to these uncertainties in traffic operations \cite{tampere2007extended}. While precise knowledge of traffic parameters is imperative to obtain the optimum control for safe and efficient traffic operation, accurate information on traffic parameters is often unavailable in practice \cite{gu2022distributionally}. Therefore, local traffic control systems rely on sensor measurements to estimate real-time traffic states and parameters that can be used for traffic-responsive control generation \cite{zheng2023urban, tampere2007extended}. Moreover, coordinated and hierarchical controller structures are often preferred for network-wide control where the lower-level local controllers are connected to an upper-level central controller to ensure stable and efficient traffic flow at the network level while minimizing the implementation and computational cost \cite{li2005distributed,tsitsokas2023two}. The central controller monitors traffic behavior to identify the operational mode of the traffic and thus specifies a suitable switching command for the local controller that switches to the appropriate control set-point or policy \cite{cazorla2022ramp}. 

Unfortunately, such hierarchical control structures highly depend on receiving accurate commands from the central controller, which are susceptible to cyberattacks \cite{zheng2023urban}. For instance, \cite{reilly2016creating} illustrated how the vulnerabilities of the traffic control system can be exploited to create on-demand congestion or complex traffic patterns. Similarly, \cite{ghena2014green} showed that an adversary may 
obtain a clear passage through intersections on the freeway by manipulating traffic controllers. 
Moreover, \cite{wang2025effects} showed
that the throughput capacity of the road decreases significantly under cyberattacks on regional traffic. In \cite{dong2024evaluating}, the authors showed that both the road capacity and the traffic safety decrease in the presence of malicious autonomous vehicles. Additionally, \cite{perrine2019implications} illustrated that an adversary can launch a coordinated attack on a few traffic signals to increase the travel time delay 4.3 times, which can cost approximately \$15 per vehicle.

At the same time, there is a threat in multi-modal cyber-physical systems due to controller switching attacks \cite{ghosh2023security}. In our previous works \cite{ghosh2024cyberattack, ghosh2025assessment}, we illustrated that switching attacks on a connected autonomous vehicle platoon can lead to unsafe traffic operations with a high risk of rear-end collision or vehicle halt on the highway. However, the threat of switching attacks has only been explored for microscopic traffic systems. Thus, in this work, we focus on detection scheme formulation against the compromised switching of the ramp metering controller for macroscopic traffic infrastructures.

\subsection{Related works}
Modern traffic control systems inherit a variety of potential security threats and attack surfaces such as signal controllers, ramp meters, sensors, or malicious vehicles \cite{roy2023redundancies}. Consequently, the authors first formulated an optimization-based attack model for ramp meters and subsequently proposed a distributed observer-based detection scheme in \cite{menelaou2024stealthy}. Similarly, the authors derived analytical conditions to realize a stealth attack and also proposed a detection framework for detectable cyberattacks in \cite{roy2020secure}. In \cite{canepa2013spoofing}, authors utilized the mixed integer linear programming technique to identify data inconsistency under a spoofing cyberattack on probe vehicle-based traffic monitoring data. On the other hand, the authors adopted Gaussian processes (GP) to detect the presence of malicious vehicles that attempt to subtly change traffic parameters, e.g., traffic density \cite{kashyap2021detection}. In addition, the authors utilized physical sensor data and social data from the drivers and passengers to detect malicious vehicles, cyberattacks on ramp metering, and corrupt sensor data \cite{roy2023redundancies}. Moreover, in \cite{adelipour2025vulnerability}, the authors adopted the multiparty computation approach to develop a secret reconstruction integrated encryption scheme that can detect and recover from a cyberattack in traffic control systems.

\subsection{Research gap and contribution}
The aforementioned works consider cyberattack detection for a homogeneous traffic operation.
However, such frameworks fail to fully capture the real-world traffic behaviors \cite{shi2021physics}. Furthermore, the inherent uncertainties of driving patterns and traffic conditions are frequently unaddressed in previous works. Thus, these works do not address the possibility and associated challenges of detecting the compromised switching of local traffic controllers. 
To address the research gaps, the contributions of this work are as follows.
\begin{enumerate}
    \item We proposed a detection scheme for compromised switching of a multi-mode inlet ramp metering controller for freeway traffic, using a multimodal macroscopic traffic model. 
    \item The proposed scheme is constructed using a bank of detectors corresponding to each admissible traffic mode to compensate for uncertain driving patterns.
    \item We derived linear matrix inequalities (LMI) based design conditions for the gain of the detectors to achieve analytical performance guarantees. 
    \item We illustrated the effectiveness of the proposed scheme against two commonly occurring attacks, namely denial-of-service and false-data-injection, under realistic traffic conditions and uncertainties.
\end{enumerate}

The rest of this paper is organized as follows. Section~\ref{pf} presents the problem framework. Section~\ref{dd} introduces the proposed detection scheme along with the design criteria. Simulation results are exhibited in Section~\ref{sr}. Finally, Section~\ref{con} concludes the work.

\textbf{Notations}: In this work, 
the $L_2$ norm of a function $f$ is defined as $\lVert f\rVert_2^2 = \int_{0}^{\infty} \!\! {f(s)^2}ds$.  For a function $f(x,t)$,  $f_t = \frac{\partial f}{\partial t}$ and $f_x = \frac{\partial f}{\partial x}$. We define an integral operator $\mathcal{I}_{A} (f)$ for any kernel $A(x,z)$ and function $f(x,t)$ as $\mathcal{I}_A ({f}) = \int_{x}^{L} \! A(x,z)\, f(z,t)\, dz$.

\begin{figure*}[t!]
  \centering
  \includegraphics[width=0.75\textwidth]{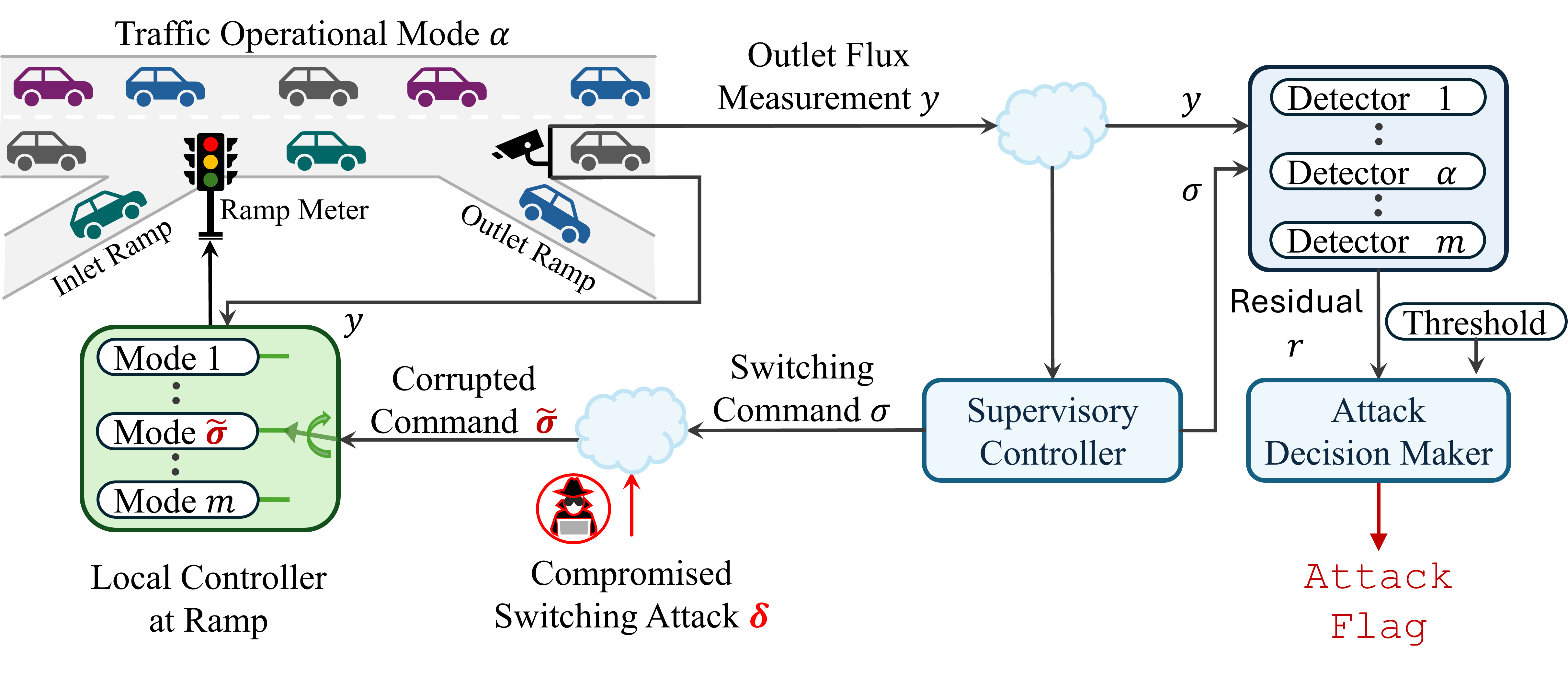}
  \caption{An overview picture for detecting switching cyberattacks on a ramp controller during multi-modal traffic operation.}
  \label{fig:Cybersecurity_scenario}
\end{figure*}
\section{ Problem Framework}\label{pf}
In this section, we specify the framework with a description of the system and its vulnerabilities. 

\subsection{Macroscopic traffic model under cyberattack}
In this work, we consider a freeway segment of length $L$ with an inlet ramp-metering at the beginning (i.\,e., $x$ = 0) and traffic sensors on the exit ramp at the end (i.\,e., $x$ = $L$).   In addition, we consider multiple modes $\alpha : [0,\infty) \rightarrow \mathcal{M} := [1, 2, \cdots, m]$ for the freeway traffic dynamics such that each mode $\alpha$ represents a distinct curve on the fundamental diagram \cite{fan2013comparative} and $m$ is the total number of admissible traffic modes. We employ the second order hyperbolic PDE Aw-Rascle-Zhang (ARZ) model \cite{aw2000resurrection} to realize this multi-modal macroscopic traffic dynamics for $x \in [0,L], \, t \in [0, \infty)$ as follows.
\begin{align}
    & q_t = - v q_x +\frac{q (\gamma p - v)}{v} v_x + \frac{q (v_f^\alpha - p - v)}{\tau v} +\eta_q, \label{flux} \\
    & v_t = (\gamma p - v) v_x +\frac{v_f^\alpha - p - v}{\tau} + \eta_v,  \label{velocity} \\
    & q(0)  =  q_s + U_{\sigma} + \delta, \quad v(L) = \frac{q (L)}{\rho^\alpha} ,\,\quad  y = q(L). \label{BC_nonlinearARZ}
\end{align} 
  Here, \(q(x,t)\), \(v(x,t)\), and $\rho (x,t)$ are respectively traffic flux, velocity, and density, and are continuously differentiable in $x \in (0,L), \, t \in [0, \infty)$.
 In particular, $\rho$ is the number of vehicles per unit length, and $q$ is the number of vehicles crossing a given point on the road per unit time such that $q = \rho v$. \eqref{BC_nonlinearARZ} represents the  the boundary conditions and measurement $y(t)$.
The traffic pressure \(p\) is an increasing function of  $\rho$ as:
\begin{align}
    p = v_f^\alpha \left( \frac{\rho}{\rho_m} \right)^{\gamma} =\frac{v_f^\alpha}{\rho_m^{\gamma}} \left( \frac{q}{v} \right)^{\gamma}, \quad \text{where} \,\,\gamma \in \mathbb{R}^+. \label{nonlinearP}
\end{align}
\(\rho_m\) is the maximum density, \( \tau \) is the relaxation time related to driving behavior,  $q_s$ is the mainline flux entering the freeway segment, and  $\eta_q, \,\eta_q \in \mathbb{R}$ capture the system uncertainties. 
$v_f^\alpha$ and $\rho^\alpha$ are the free flow velocity and desired density during traffic mode $\alpha$.

Similarly, the ramp meter also runs with different modes $\sigma \in \mathcal{M}$, where $\sigma$ is the piecewise constant switching command and obeys the average dwell time switching regulation \cite{Hespanha2003, song2025mode}. For a traffic mode $\alpha\in \mathcal{M}$, the supervisory controller chooses the corresponding ramp meter mode $\sigma=\alpha$ that is designed to appropriately control the inlet flux of the freeway traffic (in mode $\alpha$) and obtain smooth traffic with attenuated in-domain oscillations \cite{yu2019traffic}. Thus, the ramp controller output for mode $\sigma$ can be designed as  $U_{\sigma} = k^\sigma ( y - q^\sigma_\star)$, where  $ y - q^\sigma_\star$ is the measurement feedback and $k^\sigma$ is the controller gain.
The latter is designed such that the traffic dynamics settles to the steady-state points \(( q^\alpha_\star , v^\alpha_\star )\) during the traffic mode $\alpha$.

However, this operational traffic mode $\alpha$ remains uncertain due to unforeseen traffic attributes \cite{shi2021physics}. Hence, a supervisory controller monitors the sensor data to determine the operational traffic mode $\alpha$ and generates a switching command  $\sigma$ to choose the appropriate mode for the ramp meter \cite{Hespanha2003}. Data from the outlet ramp sensors to the supervisory controller and the switching command from the supervisory controller to the ramp meters are transmitted over communication networks. In this work, we consider that the adversary can corrupt the switching command $\sigma$ via the network 
and the adversarial switching command can be defined mathematically as $\widetilde{\sigma} \in \mathcal{M}\setminus \sigma,$ such that $\tilde{\sigma}=\sigma+\Delta_{\sigma}$ and $\Delta_{\sigma} \in [-(m-1), \cdots,  m-1]$. Thus, during a cyberattack, the ramp meter runs with an inaccurate mode $\tilde{\sigma}$, and the inlet flux of the freeway segment is controlled with $U_{\tilde{\sigma}}$ instead of $U_{\sigma}$. Essentially, this $U_{\tilde{\sigma}}$ will force the traffic dynamics away from the desired steady-state points \(( q^\alpha_\star , v^\alpha_\star )\). We rewrite the compromised control input $U_{\tilde{\sigma}} = U_{\sigma} + \delta$, where the switching attack $\delta$ is an unknown time function and $\delta \neq 0$ for $\Delta_{\sigma} \neq 0$ in the presence of cyberattacks and vice versa. 
Fig~\ref{fig:Cybersecurity_scenario} illustrates our problem framework of multi-modal traffic operation under the switching cyberattacks on the ramp controller.

Now, we linearize the nonlinear hyperbolic ARZ model \eqref{flux}-\eqref{BC_nonlinearARZ}  near steady-state points \(( q^\alpha_\star , v^\alpha_\star )\) for each mode $\alpha$ with small perturbations as $\tilde{q}(x,t) \!=\! q(x,t)\! -  q^\alpha_\star 
, \,\, \tilde{v}(x,t) \!=\! v(x,t)\! -  v^\alpha_\star .$
We introduce change of variables as \cite{yu2019traffic}:
\begin{align}
   & W  = f\left(\tilde{q}, \tilde{v}\right) =  e^{\left(\frac{x}{\tau  v^\alpha_\star}\right)} \left[ \tilde{q} - \frac{ q^\alpha_\star h^\alpha}{\gamma  p^\alpha_\star} \, \tilde{v}\right],  \label{newVariableW} \\
   & V = g\left(\tilde{v}\right) =  \frac{ q^\alpha_\star}{\gamma  p^\alpha_\star} \tilde{v}, \,\,\, \text{for} \,\,\,p^\alpha_\star = \frac{v_f^\alpha}{\rho_m^{\gamma}} \left( \frac{q^\alpha_\star}{v^\alpha_\star} \right)^{\gamma}.  \label{newVariableV}
\end{align}
\begin{figure*}[h]\centering
\, \subfloat[][]
{\label{fig.srCarNom}\includegraphics[width=.46\linewidth]{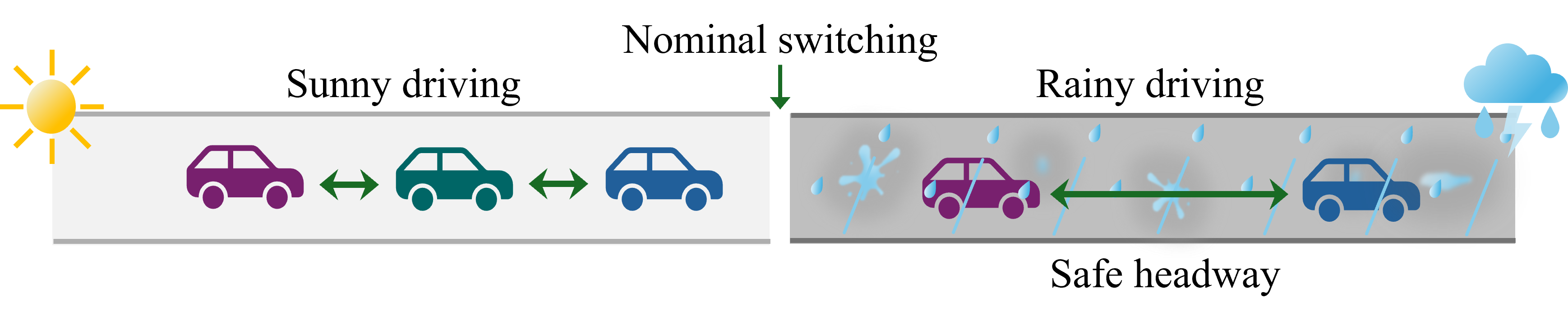}}
\,\subfloat[][]
{\label{fig.srCarDoS}\includegraphics[width=.46\linewidth]{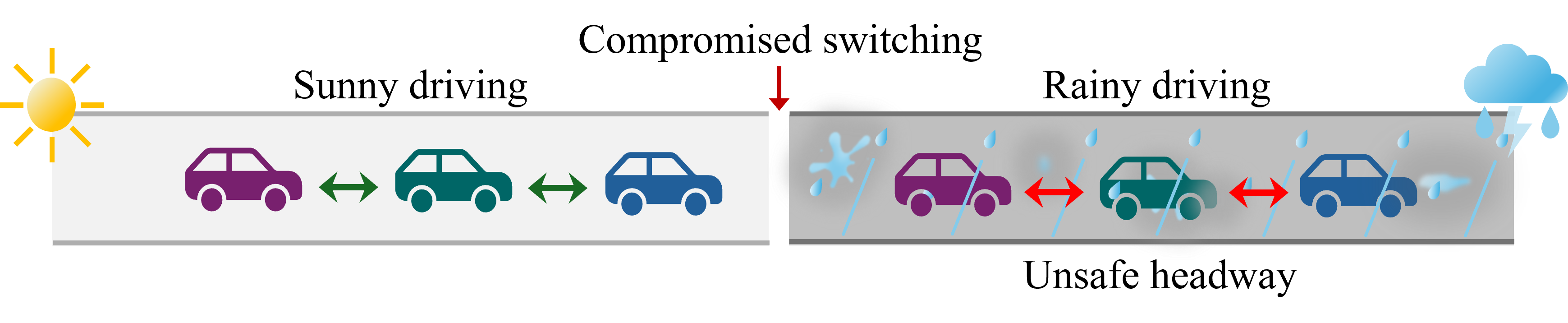}} \\\vspace{-3mm}
\subfloat[][]
{\label{fig.sr4subNom}\includegraphics[trim=0pt 0pt 00pt 0pt, clip, width=.49\linewidth]{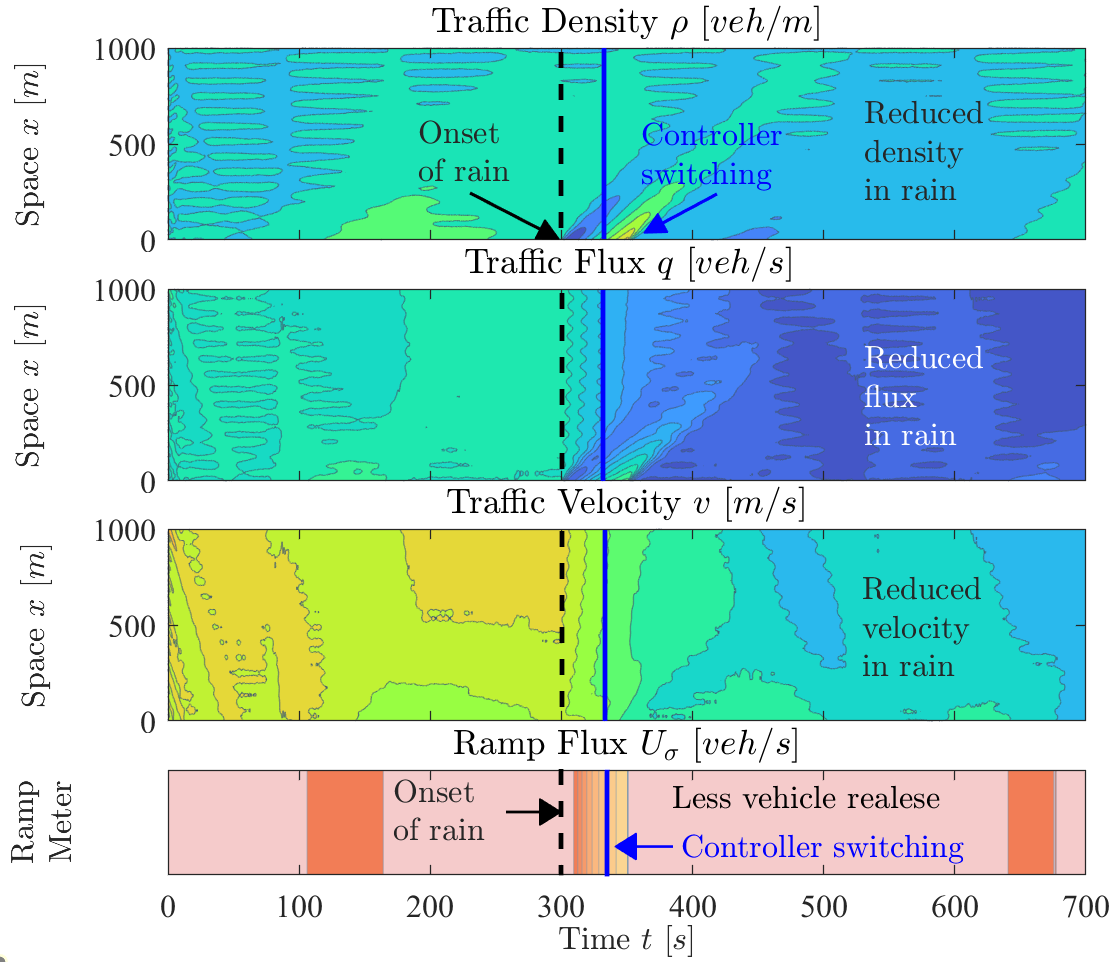}}
\hfill
\subfloat[][]
{\label{fig.sr4subDoS}\includegraphics[trim=00pt 0pt 0pt 0pt, clip, width=.49\linewidth]{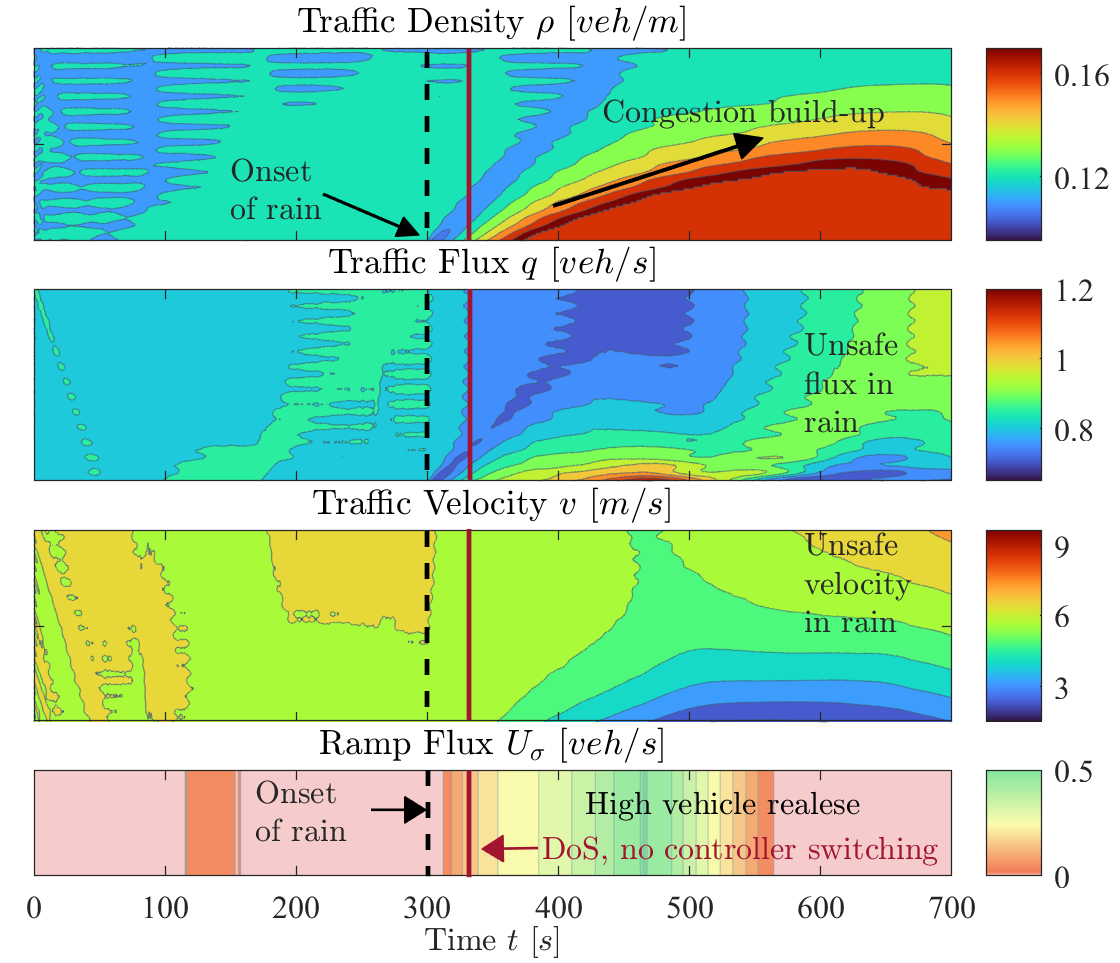}}
\caption{Figure depicts the impact of controller switching on traffic behavior under nominal conditions and DoS attack. (a) Illustrates the safe headway extended from sunny to rainy weather under nominal switching, and (c) shows reduced traffic density, flux, and velocity under adjusted ramp metering that compensates for the adverse weather. Conversely, (b) shows the unsafe headway under a wrong controller in rainy weather, while (d) demonstrates the congestion build-up along with unsafe traffic flux and velocity due to unadjusted ramp metering under DoS attack, i.e., no controller switching.}\label{fig.impact}
\end{figure*}
Next, we utilize the definitions \eqref{newVariableW}- \eqref{newVariableV} to map the \eqref{flux}-\eqref{BC_nonlinearARZ} system to the linear hyperbolic system: 
\begin{align}
    & W_t (x,t) = -  v^\alpha_\star W_x (x,t) + \eta_W, \label{omegaF} \\
    & V_t (x,t) =  v^\alpha_\star h^\alpha V_x (x,t) + \beta^\alpha (x) W (x,t) + \eta_V, \label{nuF} \\
    & W (0,t) = - h^\alpha V (0,t) + q_s + U_\sigma -  q^\alpha_\star + \, \delta,  \label{BC_omegaNUF} \\
    & V (L,t) = l^\alpha W (L,t), \quad \, \, \overline{y} (t) = c^\alpha W (L,t). \label{y_omegaNUF}
\end{align}
Here, $ \eta_W = f \left(\eta_q, \eta_v\right)$, $\eta_V =  g\left(\eta_v\right)$, $h^\alpha = \frac{\gamma  p^\alpha_\star -  v^\alpha_\star}{ v^\alpha_\star}$, $\beta^\alpha(x) = \frac{-e^{\frac{-x}{\tau  v^\alpha_\star}}}{\tau} $, $ l^\alpha =   e^{\frac{-L}{\tau  v^\alpha_\star}}$,  $c^\alpha = \frac{\gamma  p^\alpha_\star}{ v^\alpha_\star}l^\alpha$, and $\overline{y} = y -  q^\alpha_\star $. Next, we present the cyberattack policies for our framework and mathematically define cyberattack $\delta$.

\subsection{Cyberattack policies and their impact} \label{impact_sec}
In this framework, we consider that the adversary deliberately compromises the switching command generated by the supervisory controller to achieve their malevolent objectives, such as large-scale and unsafe traffic disruptions, on-demand congestion, or selective traffic advantages (e.g., securing a clear passage) \cite{reilly2016creating, ghena2014green}. They may activate a wrong controller to lower the traffic influx from the ramp than desired to decrease the throughput capacity of the freeway. Alternatively, the adversary can increase the traffic influx from the ramp than desired,  which will result in lower and unsafe vehicle headway \cite{shi2021constructing}. The adversary may also manipulate the traffic influx to generate specific scenarios such as bottle-neck or stop-go traffic. Moreover, the adversary may opt for several attack policies such as denial-of-switching (DoS), false-data-injection (FDI), or replay attack \cite{ghosh2024cyberattack}.

Let us consider Fig~\ref{fig.impact}. Adverse weather conditions such as rain, snow, and fog significantly decrease free-flow velocity due to reduced visibility and adhesion \cite{nigam2022empirical}.  Consequently, in rainy weather, vehicles run slower while keeping a larger gap ahead, as shown in Fig~\ref{fig.srCarNom}. This leads to a reduction in traffic influx and velocity. The black dotted line on Fig~\ref{fig.sr4subNom} marks this change in the driving pattern. 
The supervisory controller identifies this altered traffic behavior within 40s and generates the appropriate switching command to ensure traffic safety. 
Consequently, the local controller switches mode accordingly, such that the ramp meter is adjusted to release fewer vehicles from the ramp and compensate for such adverse conditions. This delayed nominal switching is indicated by the blue line on Fig~\ref{fig.sr4subNom}, and the bottom plot shows the adjusted ramp metering. This ensures safe traffic operation with reduced density, flux, and velocity, as shown in the top three plots of Fig~\ref{fig.sr4subNom}.

\vspace{-2mm}During a DoS attack, the adversary may prevent the local controller from switching to the appropriate mode in response to the onset of rain and thus force the ramp meter to release more vehicles from the ramp. The dark red line on Fig~\ref{fig.sr4subDoS} indicates this DoS attack, which leads to the unadjusted ramp metering shown in the bottom plot. This results in a congestion build-up as seen from the red contour on the top plot of Fig~\ref{fig.sr4subDoS}. Moreover, this DoS attack may result in unsafe traffic operation with higher velocity and hazardously small headway in rainy weather. Fig~\ref{fig.srCarDoS} and Fig~\ref{fig.sr4subDoS} capture this scenario.  Next, we introduce our proposed cyberattack detection scheme to detect the controller switching attack $\delta$.

\vspace{-3mm}
\section{Cyberattack Detection Scheme}\label{dd}

\subsection{Detector and error dynamics}
\vspace{-3mm}
In this work, we propose a detection scheme that can reliably detect the presence of corrupted switching with analytical guarantees. In this framework, the operational mode of the traffic is uncertain and is estimated by the supervisory controller for choosing the correct local controller mode. Therefore, as the operational mode of the traffic switches, there will be a delay in identifying this system mode. Now, if a single detector is used such that its mode of operation ``matches" the traffic mode, reliable detection of attacks will be dependent on the instantaneous identification of the correct system mode by the supervisory controller. Any temporary loss of synchronicity between the operational traffic mode and the mode of the detector will be flagged as an anomaly in that scenario.
Therefore, to isolate the effect of switching attacks from anomalies due to system switching, we propose a bank of $m$ detectors corresponding to the $m$ admissible traffic modes for the freeway. In particular, we propose this detector bank based on the output injection-based residual generation approach. 
To realize this proposed structure, we define the dynamics of the $j^{th}$ detector $\forall j \in \mathcal{M}$ based on the linear hyperbolic traffic dynamics \eqref{omegaF}-\eqref{y_omegaNUF} as:
 \begin{align}
    & \hat{W}^j_t (x,t) = -  v^j_\star \hat{W}^j_x (x,t) + k_1^j (x) \zeta^j, \label{omegaHat} \\
    & \hat{V}^j_t (x,t) =  v^j_\star h^j \hat{V}^j_x (x,t) + \beta^j \hat{W}^j (x,t) + k_2^j (x) \zeta^j, \label{nuHat} \\
    & \hat{W}^j (0,t) = - h^j \hat{V}^j (0,t) + q_s + U_\sigma -  q^j_\star,  \label{BC_omegaNUhat} \\
    &  \hat{V}^j (L,t) = l^j \hat{W}^j (L,t) + k_3^j (x) \zeta^j. \label{y_Hat} 
\end{align}
Here, for the $j^{th}$ detector, $\zeta^j \!= \!\overline{y} \!- \! c^j \hat{W}^j (L,t)$ is the output, $k_1^j$, $k_2^j$, and $k_3^j$ are the detector gains to be designed. 
Now, the errors in $W$ and $V$ are given by $ \omega^j = W - \hat{W}^j$ and $ \vartheta^j = V - \hat{V}^j$, respectively. Then, we can obtain the error dynamics for the $j^{th}$ detector from \eqref{omegaF}-\eqref{y_omegaNUF} and \eqref{omegaHat}-\eqref{y_Hat} as:
 \begin{align}
    & \omega^j_t\! (x,t)\! = \!-  v^\alpha_\star \omega^j_x \!(x,t)\!+ \!\eta_W \!  + \!\theta_v  \hat{W}^j_x \!(x,t)\! - \!k_1^j (x) \zeta^j, \label{omegaErr} \\
    & \vartheta^j_t (x,t) =  v^\alpha_\star h^\alpha \vartheta^j_x (x,t) + \beta^\alpha (x) \omega^j (x,t) +  \eta_V \nonumber \\
    & \quad \qquad  + \theta_{h,v} \hat{V}^j_x (x,t)  + \theta_\beta \hat{W}^j (x,t)  - k_2^j (x) \zeta^j, \label{nuErr}  
    \end{align}
    with the boundary conditions and output obtained as
     \begin{align}
    & \omega^j (0,t) = - h^\alpha \vartheta^j (0,t) + \theta_h \, \hat{V}^j (0,t) - \theta_q + \delta,  \label{BC_omegaErr} \\
    & \vartheta^j (L,t) = l^\alpha \omega^j (L,t) + \theta_l \, \hat{W}^j (L,t) - k_3^j (x) \zeta^j, \label{BC_omegaErr2} \\
    & \zeta^j = c^\alpha \omega^j (L,t) + \theta_c  \hat{W}^j (L,t). \label{yErr}
\end{align}
Finally, we define the residual of the detection scheme as the minimum among the detector outputs, i.e.
\begin{align}
    r(t)\! =\! \min\limits_j \lvert \zeta^j (t)\rvert ,\! \text{ which implies }{\lVert r \rVert}^2_2 \leqslant {\lVert \zeta^j \rVert}^2_2, \, \forall j. \label{residual}
\end{align}
This residual $r(t)$ is used to detect the presence of the controller switching attack $\delta$, and an attack decision is made only if and when the residual $r$ crosses a predefined threshold $J$. The threshold ${J}$ can be defined by fixing the false alarm rate  i.\,e., the probability $P_{D}$  of the residual  $r$ crossing the threshold ${J}$  under nominal operation \cite{Ding}. 
Mathematically, we can write
\begin{align}
    P_{D} & = P\left( r > {J} | \delta  = 0 \right)  = \int_{J}^{\infty}\!\!\! P(r_{\eta}) \,dr_{\eta}, \label{threshold}
\end{align}  
where $P(r_{\eta})$ is the probability distribution of the generated residual data $r_\eta$. For instance, to achieve a false alarm rate of 15\% we can set the threshold ${J}$ such that $P_{D} < 0.15$. 
For the proposed structure, one detector from the bank of $m$ detectors always matches the traffic dynamics (even under delayed traffic mode identification), and that $(j= \alpha)^{th}$ detector generates the least output \cite{ghosh2023security}. Thus, this bank of detector structure ensures that the residual is small enough when $\delta =0$. Contrarily, under corrupted ramp metering, all detectors run with the accurate $U_{\sigma}$, while the traffic dynamics runs with compromised boundary control $U_{\tilde{\sigma}} = U_{\sigma} + \delta$. This incoherence leads to a higher residual $r$  for $\delta \neq 0$ and can be used for attack detection using threshold crossing.
Fig. \ref{fig:Cybersecurity_scenario} illustrates the proposed structure of the detection scheme with a bank of $m$ detectors and an attack decision maker. Next, we present the mathematical analysis to obtain the LMI-based design criteria for the proposed detection scheme.

\subsection{Analytic guarantees for detector performance}
We consider four performance criteria for the residual $r(t)$ to guarantee the performance of the detection scheme. They are: nominal exponential stability (ES), anomaly/uncertainty-to-residual stability (AURS), robustness against uncertainties, and sensitivity to cyberattacks. These criteria are defined as follows \cite{Troy_actuator_anomaly}.

\textbf{ES:} Under nominal traffic control with no uncertainty (\mbox{$\delta, \eta_q, \eta_v = 0$}) and design parameters $\Upsilon_1, \Upsilon_2 \!\in \!\mathbb{R}^+$ 
\begin{align}
    r^2(t) \leqslant \Upsilon_2 \,r^2(0) \, e^{-\Upsilon_1 t}. \label{ES_criteria}
\end{align}
\textbf{AURS:} Under compromised traffic control and uncertainties (\mbox{$\delta, \eta_q, \eta_v \neq 0$}) and for design parameters $\Upsilon_2,\Upsilon_3,\Upsilon_4 \!\in \! \mathbb{R}^+$
\begin{align}
    r^2(t) & \leqslant \Upsilon_2 \,r^2(0) \, e^{-\Upsilon_3 t} \nonumber \\
    & \,\, + \Upsilon_4 \sup\limits_{t\geqslant0} \left[ \delta^2 + \lVert \eta_v \rVert^2 + \lVert \eta_q \rVert^2 + \eta_q^2(L,t) \right ]. \label{DSS_criteria}
\end{align}
\textbf{Robustness:} Under nominal traffic control ($\delta= 0$) and uncertainties {$\eta_q, \eta_v \! \neq \!0$}, and design parameter $\!\Upsilon_5 \! \in \!\mathbb{R}^+$
\begin{align}
    \int_0^\infty \!\!\!\!\!\! r^2 dt \leqslant \Upsilon_{5} \int_0^\infty \!\!\!\!\!\!  \left[ \lVert \eta_v \rVert^2 + \lVert \eta_q \rVert^2 + \eta_q^2(L,t) \right ] dt + \epsilon. \label{robust def r}
\end{align}
\textbf{Sensitivity:} Under compromised traffic control $\delta \! \neq \!0$, no uncertainty ({$ \eta_q,\! \eta_v \!\!= \!\!0$}), and design parameter $\Upsilon_6 \!\in \!\mathbb{R}^+$
\begin{align}
    \int_0^\infty r^2 dt \geqslant \Upsilon_{6} \int_0^\infty \delta^2 dt - \epsilon. \label{sen def r}
\end{align}
In particular, the criteria \eqref{ES_criteria}-\eqref{sen def r} ensures the following.
\begin{enumerate}
    \item \textit{Exponential convergence} of the residual $r(t)$ under nominal operations, no uncertainties ($\delta, \eta_q,\eta_v = 0$).
    \item \textit{Boundness} of the detectors in all circumstances with bounded disturbances, including $\delta, \eta_q,\eta_v \neq 0$.
    \item \textit{Minimized false alarms} as the residual $r(t)$ stays below the threshold during nominal traffic operations under uncertainties ($\delta =0, \eta_q,\eta_v \neq 0$).
    \item \textit{Reduced misdetection} as the residual $r(t)$ crosses the threshold under a controller switching attack without uncertainties ($\delta \neq 0, \eta_q,\eta_v = 0$).
\end{enumerate}
Now, we will derive LMI-based conditions for detector gains $k_1^j, $ $ k_2^j$, and $k_3^j, \forall j \in \mathcal{M}$ to ensure desired performance for the detection residual $r$ \eqref{residual} such that \eqref{ES_criteria}-\eqref{sen def r} are satisfied. To achieve this, we first apply the backstepping strategy to decouple the coupled error dynamics \eqref{omegaErr}-\eqref{BC_omegaErr2}.

\subsection{Backstepping transformation of error dynamics}
Using standard methodology \cite{roy2020secure}, we obtain the backstepped system \eqref{phi_tx}-\eqref{BC_psi_l} that decouples the error dynamics \eqref{omegaErr}-\eqref{BC_omegaErr2}.
\begin{align}
    & \phi^j_t (x,t) = -  v^\alpha_\star \phi^j_x (x,t)+ a^j (x,t) + \theta_\phi (x,t),   \label{phi_tx} \\
    & \psi^j_t (x,t) =  v^\alpha_\star h^\alpha \psi^j_x (x,t) + b^j (x,t) + \theta_\psi (x,t),  \label{psi_tx} \\
    & \phi^j (0,t) = - h^\alpha \psi^j (0,t) + \theta_h \, \hat{V}^j (0,t) - \theta_q + \delta,  \label{BC_phi_zero} \\
    & \psi^j (L,t) = {\lambda}^j  \phi^j (L,t) + \theta_k \, \hat{W}^j (L,t). \label{BC_psi_l}
\end{align}
The backstepping transformation is given by 
\begin{align}
     & \omega^j(x,t) = \,\,\, \mathcal{R}^j (\phi)\,\,\, = \phi^j(x,t) - \mathcal{I}_{R^j} ( \phi^j), \label{phi2omega} \\
     & \vartheta^j (x,t) =  \mathcal{S}^j (\psi, \phi)  = \psi^j(x,t) \! - \mathcal{I}_{S^j} ( \phi^j). \label{psi2v}
\end{align}
Here, $a^j $ and $b^j $ is related to the uncertainties as $ \eta_W = \mathcal{R}^j (a^j)$ and $\eta_V = \mathcal{S}^j (b^j, a^j)$. Similarly,
$\theta_v \, \hat{W}^j_x (x,t) -  \theta_c \, k_1^j  \hat{W}^j (L,t)=  \mathcal{R}^j ( \theta_\phi)$ and $ \theta_{h,v} \hat{V}^j_x (x,t)  + \theta_\beta \hat{W}^j (x,t)   -  \theta_c \, k_2^j  \hat{W}^j (L,t) = \mathcal{S}^j (\theta_\psi, \theta_\phi)$. 
We can rewrite $\zeta^j$ as
\begin{align}
    & \zeta^j =  c^\alpha \phi^j (L,t) + \theta_c  \hat{W}^j (L,t). \label{Err_phi_psi}
\end{align}
The kernel functions in \eqref{phi2omega}-\eqref{psi2v} can be derived to be
\begin{align}
    R^j(x,z) &=  \frac{h^\alpha}{\tau \gamma  p^\alpha_\star} \exp{\left(\frac{h^\alpha \left( x-z\right)}{\tau \gamma  p^\alpha_\star}\right)}, \label{R_def}\\
    S^j (x,z) & = - \frac{1}{\tau \gamma  p^\alpha_\star} \exp{\left(\frac{- x- h^\alpha z}{\tau \gamma  p^\alpha_\star}\right)}. \label{S_def}
\end{align}
Then, the detector gains $k_1^j$ and $k_2^j$ can be expressed in terms of kernels as 
\begin{align}
    k_1^j (x)  & = \frac{ v^\alpha_\star}{c^\alpha} R^j(x,L), \,\quad k_2^j (x)   = \frac{ v^\alpha_\star}{c^\alpha} S^j(x,L).
\end{align}
The detector gain $k_3^j$ is related to $\lambda^j$ \eqref{BC_psi_l} as
 \begin{align}
    {\lambda}^j = l^\alpha - k_3^j. \label{lambda}
\end{align}
Furthermore, the backstepping transformations \eqref{phi2omega}-\eqref{psi2v} are invertible using the transformations \eqref{omega2phi}-\eqref{v2psi}.
\begin{align}
    & \phi^j(x,t) = \mathfrak{M}^j (\omega^j) =  \omega^j(x,t) - \mathcal{I}_{M^j} ( \omega^j), \label{omega2phi} \\
    & \psi^j (x,t) = \mathfrak{N}^j (\vartheta^j,\omega^j) = \vartheta^j(x,t) - \mathcal{I}_{N^j} ( \omega^j). \label{v2psi}
\end{align}
Using these maps \eqref{omega2phi}-\eqref{v2psi} and comparing the error dynamics \eqref{omegaErr}-\eqref{BC_omegaErr2} with the backstepped dynamics \eqref{phi_tx}-\eqref{BC_psi_l}, we obtain the relations bellow
\begin{align}
        &  a^j (x,t) =   \mathfrak{M}^j (\eta_q ), \qquad \quad  b^j (x,t) =    \mathfrak{N}^j (\eta_v, \eta_ q), \label{a2dq} \\
    & \theta_\phi^j (x,t)  = \mathfrak{M}^j (\theta_v \hat{W}^j_x  ) -  \theta_c \mathfrak{M}^j (k_1^j ) \hat{W}^j (L,t), \label{thetaPhiker} \\   
    & \theta_\psi^j (x,t) = \mathfrak{N}^j( \theta_{h,v} \hat{V}^j_x (x,t),    \theta_v \hat{W}^j_x) + \theta_\beta \hat{W}^j (x,t)  \nonumber \\
    &\qquad \quad - \theta_c[k_2^j     -  \mathcal{I}_{N^j} ( k_1^j)] \hat{W}^j (L,t). \label{thetaSIkernel} 
\end{align}
Furthermore, the kernels in  \eqref{omega2phi}-\eqref{v2psi} are given by 
\begin{align}
    M^j(x,z) = - R^j(x,z), \quad N^j(x,z) = - S^j(x,z). \label{MN_def}
\end{align} 
We note that the detailed derivation of \eqref{phi_tx}-\eqref{BC_psi_l} has been omitted for lack of space. 

\subsection{Design of detector gains}
Since the initial errors in states for the detectors will always remain bounded, we will assume the following. 
\begin{assm} \label{IC_assm}
    We assume that the initial conditions for the backstepped system \eqref{phi_tx}-\eqref{BC_psi_l} are bounded as 
    \begin{align}
     0 \!\leqslant \underline{\phi}{^j_0}^2 (x) \!\leqslant  \phi{^j}^2 (x,0)\! \leqslant \overline{\phi}{^j_0}^2 (x) \!\leqslant \!\infty, \,\,\, \forall x \in [0,L]. \label{ICBC} 
    \end{align}
\end{assm}
Using the second order decoupled linear PDE \eqref{phi_tx}-\eqref{BC_psi_l} and Assumption~\ref{IC_assm}, we now present our main Theorem that provides the design conditions on gains $k_1^j, $ $ k_2^j$, and $k_3^j, \forall j \in \mathcal{M}$ as LMIs, to ensure reliable detection performance with regards to criteria \eqref{ES_criteria}-\eqref{sen def r}.

\begin{thmm} \label{DSS_th} 
    Consider a multi-modal freeway traffic modeled by \eqref{flux}-\eqref{BC_nonlinearARZ} and the corresponding bank of $m$  detectors \eqref{omegaHat}-\eqref{y_Hat}. The detection scheme is considered ES \eqref{ES_criteria}, AURS \eqref{DSS_criteria}, robust against uncertainties \eqref{robust def r}, and sensitive towards cyberattacks \eqref{sen def r} for every mode $\alpha$, if there exists a  constant $\xi\in \mathbb{R^+}$ such that $\forall \alpha,j \,\in \mathcal{M}$:
    \begin{align}
         0\leqslant \mu_1  k_3^{j} & \leqslant \Upsilon_7^\alpha, \,\quad \quad \quad \quad \text{[ES \& AURS  criteria]} \label{dssThc}\\
         \Lambda_1^j & \leqslant 0, \quad \quad \quad \quad \,\,\text{[Robustness criterion]} \label{rbostThc}\\
         \Lambda_2^j & \leqslant 0, \quad \quad \quad \quad \,\,\, \text{[Sensitivity criterion]} \label{sensThc}
    \end{align}
where,
\begin{align}
    & \Upsilon_7 = \sqrt{\frac{  v^\alpha_\star - \mu_2 e^{-L}- {c^\alpha}^2e^L}{\mu_3  v^\alpha_\star\,h^\alpha\,e^{2L}}- \mu_1 e^{-\frac{2L}{\tau  v^\alpha_\star}}}, \label{upK4J}\\
    & \Lambda_1^{j} = \text{diag}\{\mu_4 \!- \!\!\overline{\Upsilon}_8^{\alpha,j}, \,-\Upsilon_9, -\Upsilon_{10}, -\Upsilon_{11}, -\Upsilon_{12} \}, \label{robLam}\\
    & \Lambda_2^{j} = \begin{bmatrix}
    \overline{\Upsilon}_8^{\alpha,j} -\Upsilon_{16} & \qquad -1 & \qquad -1 \\
    -1 &  \qquad \Upsilon_{17}    & \qquad 0 \\
    -1 &\qquad 0 &  \qquad \Upsilon_{18} 
\end{bmatrix}. \label{senLam}
\end{align}
Table \ref{tab:thm_var} lists the parameters for \eqref{upK4J}-\eqref{senLam}, where $\mu_i > 0$, $i\in \{1,\cdots,16\}$ are free tuning parameters to be designed.
\renewcommand{\arraystretch}{1.3}
\begin{table}[h!]
    \centering
            \caption{List of parameters in \eqref{upK4J}-\eqref{senLam}.}
    \begin{tabular}{|l|l|}
        \hline
        $\overline{\Upsilon}_8^{\alpha,j}$ = $ \Upsilon_{8}^{\alpha,j}\frac{ \xi}{2 e^L} $ &$\Upsilon_{8}^{\alpha,j}$ = $\frac{ v^\alpha_\star} {\xi} - \frac{ {\mu}_3  v^\alpha_\star  h^\alpha  {\lambda^j}^2}{e^{-2L} \xi} - \frac{{\mu}_2 }{e^{L}\xi}$  \\
        \hline
        $\Upsilon_{9} $ = $\Upsilon_5 - \mu_5$ &    $\Upsilon_{13}^j $ = {$ \mu_7 + L^2\frac{ \mu_7 h{^\alpha}^2 + \mu_5 }{(\tau \gamma  p^\alpha_\star)^2} $  }  \\
        \hline
         $\Upsilon_{10} $ = $\Upsilon_5 - \Upsilon_{13}^j $   &$ {\Upsilon}_{14} $ = {$ \mu_{8} + L^2\frac{ \mu_{8} h{^\alpha}^2 + \mu_{13} }{(\tau \gamma  p^\alpha_\star)^2} $  } \\
         \hline
       $\Upsilon_{11} $ = $ \Upsilon_5 -\frac{\mu_{6} \xi{^j}^2}{2}$ & $\Upsilon_{17} $ = $\mu_{12}  v^\alpha_\star  +\Upsilon_{6} + \mu_{15}$\\
        \hline
       $\Upsilon_{12} $ = $\Upsilon_5 - \mu_4 -\Upsilon_{\theta}$   & $ {\Upsilon}_{18} $ = {$ {\Upsilon_\theta}+ {\Upsilon_6} + \mu_{16} + 1$  }\\
        \hline
         \multicolumn{2}{|l|}{ $\Upsilon_\theta$ =  $\quad\max({\mu}_{10}  v^\alpha_\star,  {\mu}_{11}  v^\alpha_\star, {\mu}_{13},{\mu}_{14},{\Upsilon}_{14},\Upsilon_{15})$}\\
        \hline
        \multicolumn{2}{|l|}{ $\Upsilon_{15}$ = $\,\, {\Upsilon}_{14} \lVert k_1^j  \rVert^2 + {\mu}_{13}  \lVert k_2^j  \rVert^2 + {{\mu}_3 { v^\alpha_\star  h^\alpha e^{L} }  \theta_k^2 }{\theta_c^2}^{-1} + {\mu}_{14}$}\\
        \hline
    \end{tabular}
    \label{tab:thm_var}
\end{table}
\end{thmm}

\begin{proof}
    The proof is presented in the appendix.
\end{proof}
In the next section, we present our simulation results to evaluate the performance of the proposed detection scheme, designed based on Theorem \ref{DSS_th}.

\section{Simulation Results}\label{sr}
\begin{table}[h!]
    \centering
       \caption{List of traffic simulation parameters for each mode $\alpha$.}
    \begin{tabular}{|c|c|c|c|c|c|}
    \hline
        Mode & $v_f^\alpha $ & $\rho^\alpha $ & $ v^\alpha_\star$ & $ q^\alpha_\star$ & \multirow{2}{*}{$k^\sigma$}\\
        $\alpha $ & $ [m/s]$ & $[veh/m]$ & $ [m/s]$ & $[veh/s]$ & \\
        \hline
       (Rainy) 1 & 25 $\pm 2.5$ &  \multirow{4}{*}{0.12} & 6 & 0.72 & 1.8\\
        \cline{1-2} \cline{4-6}
        (Light) 2 & 30 $\pm 2.5$ &  & 7 & 0.84 & 2.4 \\
        \cline{1-2} \cline{4-6}
       (Heavy) 3 & 35 $\pm 2.5$ &  & 8 & 0.96 & 3.5\\
        \hline
    \end{tabular}
    \label{tab:modeParam}
\end{table}
We present case studies that demonstrate two attack scenarios and the efficacy of our proposed algorithm for detection. We consider a freeway segment with a length of $L =1\, km$ and a maximum density of $\rho_m$ = 160   vehicles/km. We take the relaxation time $\tau = 60 s$ and $\gamma = 1$ and consider sinusoidal spatial initiation for $q$ and $v$ at $t =0$. To imitate real-world traffic operations, we consider 4 sources of uncertainties:
(i) in-domain sinusoidal uncertainties $\eta_q$ and $\eta_v$ respectively in $q$ and $v$ with $\pm 0.12\%$, (ii) uncertain free flow velocity $v_f$ ($\pm 2.5 m/ s$), (iii) varied mainline flux $q_s$ at the inlet with $\pm 10\%$, and (iv) measurement noise of $\pm 2\%$ at the outlet. To illustrate our method, we consider have 3 admissible traffic modes (rainy, light, heavy) for the freeway along with corresponding 3 controller modes and a bank of 3 detectors, i.e., $m = 3$ and $\alpha, \sigma, j \in [1,2,3]$ and Table~\ref{tab:modeParam} lists the traffic parameters used in our simulations for each mode $\alpha$. The threshold is obtained as 0.02 for each traffic mode using \eqref{threshold}.

\subsection{Case study I: DoS attack}
\begin{figure}[h]
    \centering
    \includegraphics[width=0.65\linewidth]{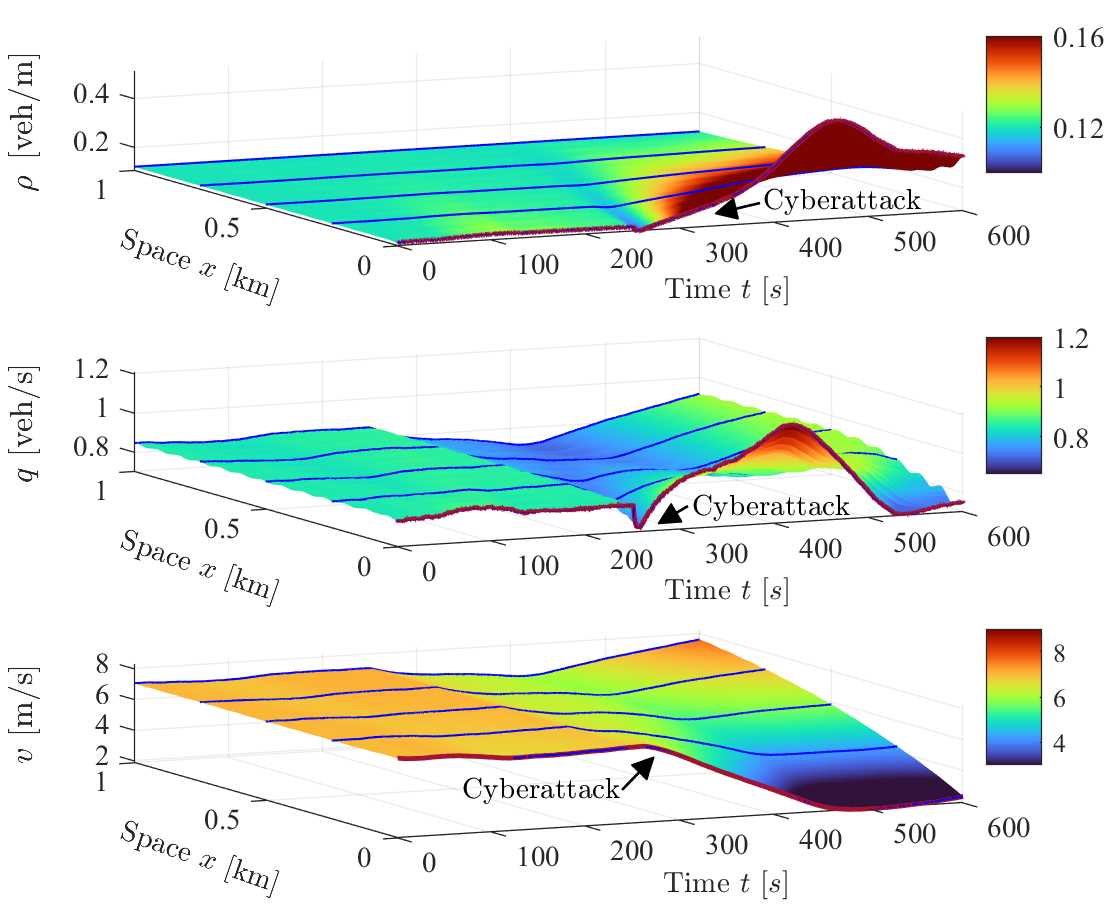}
    \caption{Under DoS attack, the figure shows traffic density (top),  flux (middle), and  velocity (bottom).}
    \label{fig:case1rqv}
\end{figure}
\begin{figure}[h]
    \centering
    \includegraphics[width=0.6\linewidth]{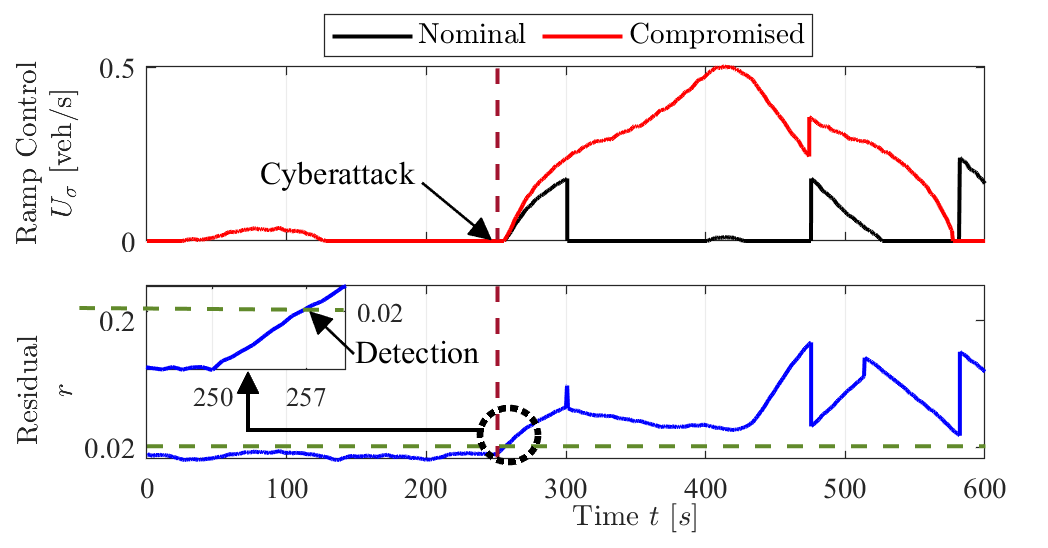}
    \caption{Under DoS attack, the figure shows nominal and corrupted boundary ramp control (top) and attack detection with the generated residual (bottom).}
    \label{fig:case1Res}
\end{figure}
For this case study, we consider the DoS cyberattack scenario described in Section~\ref{impact_sec}. Here, the traffic driving pattern changes from mode 2 to 1 (Table~\ref{tab:modeParam}) at the onset of rain. From the start of the DoS attack at 250s, the local ramp controller is unable to receive supervisory commands, and it continues to run in mode 2 (light traffic) despite the changed driving conditions due to the onset of rain. This scenario has a two-fold effect. First, at the freeway entry, this causes a higher flux caused by the release of more vehicles from the incorrect ramp controller and a lower velocity caused by adverse weather conditions -- leading to \textit{traffic congestion} (or high traffic density $\rho$). The red contour in the top plot of Fig.~\ref{fig:case1rqv} clearly shows the congestion scenario. Similarly, the middle and bottom plots of Fig.~\ref{fig:case1rqv} show the high flux and low velocity at the inlet, respectively. To counteract the high influx and low velocity of the inlet, the traffic at the outlet of the freeway reaches higher flux and higher velocity conditions, which are \textit{unsafe for rainy driving}. Thus, through this DoS attack, the adversary not only sways the traffic dynamics from its desired equilibrium of \(( q^1_\star , v^1_\star )\) but also successfully achieves both undesirable and unsafe driving outcomes.
Specifically, the top plot of Fig~\ref{fig:case1Res} illustrates the higher influx of vehicles from the ramp under the DoS corrupted control $U_{\tilde{\sigma}} = U_2$ compared to the desired lower influx under nominal control $U_{{\sigma}} = U_1$. Now, the bottom plot of Fig~\ref{fig:case1Res} shows that the generated residual $r$ crosses the threshold within 8s of attack injection, resulting in accurate detection of the  DoS attack $\delta$.

\subsection{Case study II: FDI attack}
\begin{figure}[h]
    \centering
    \includegraphics[width=0.65\linewidth]{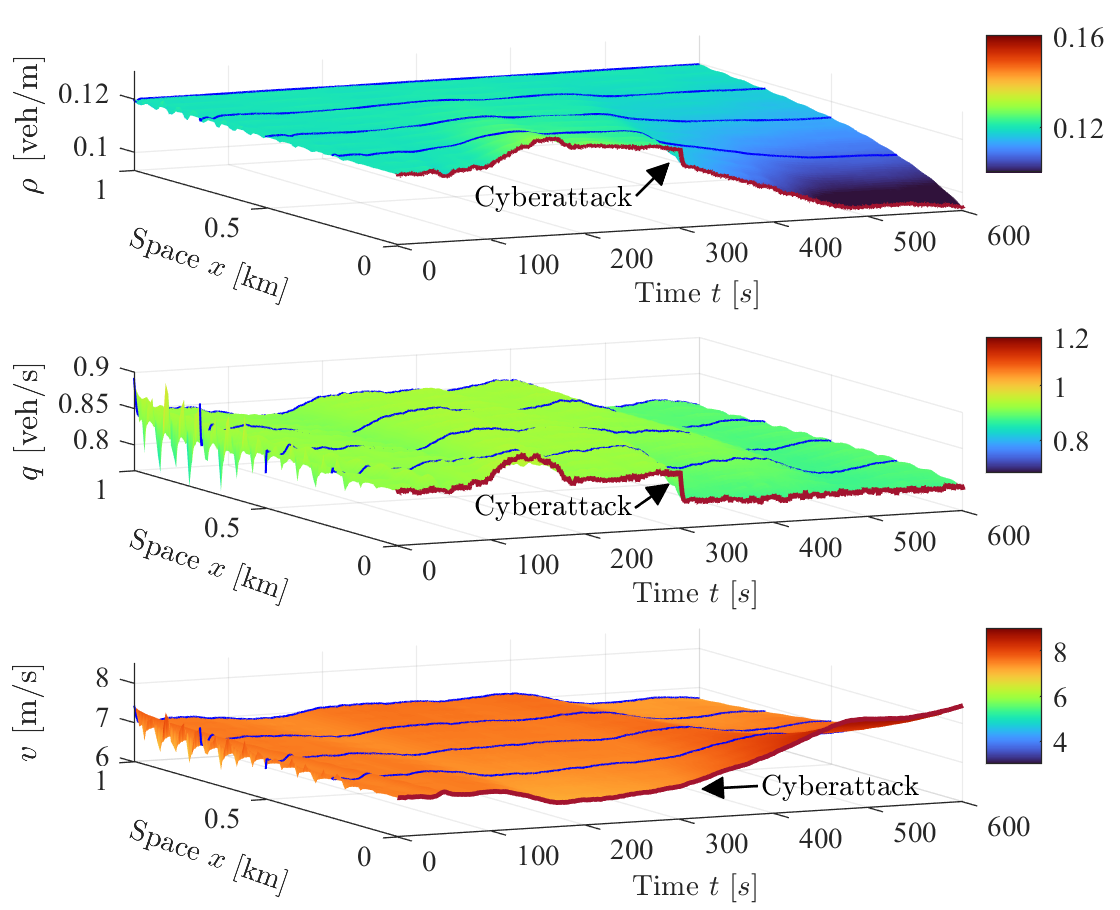}
    \caption{Under FDI attack, the figure shows traffic density (top),  flux (middle), and  velocity (bottom).}
    \label{fig:case2rqv}
\end{figure}

\begin{figure}[h]
    \centering
    \includegraphics[width=0.6\linewidth]{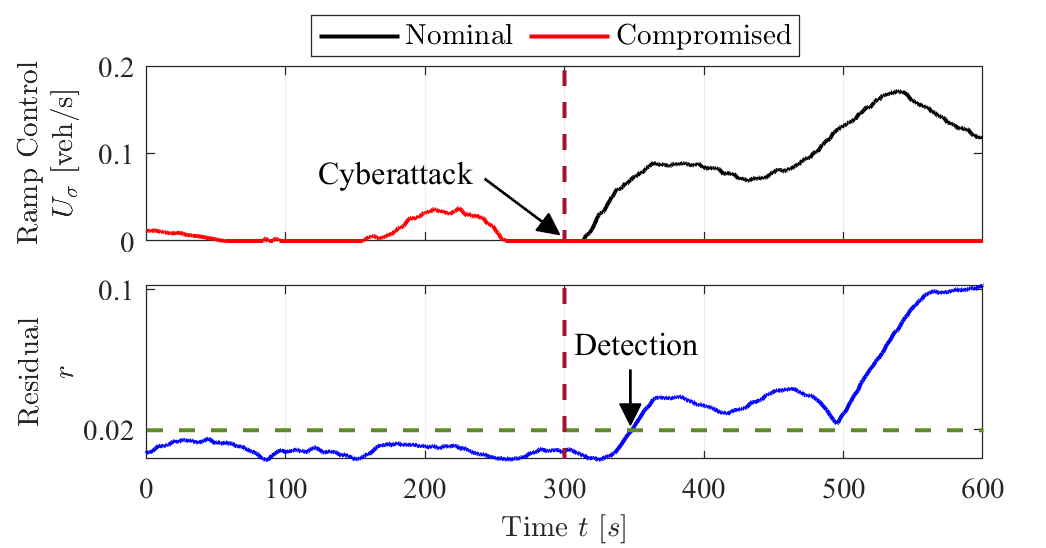}
    \caption{Under FDI attack, the figure shows nominal and corrupted boundary ramp control (top) and attack detection with the generated residual (bottom).}
    \label{fig:case2Res}
\end{figure}
In this scenario, we consider that the freeway traffic is running with mode 3 (Table~\ref{tab:modeParam}), i.e.,
heavy traffic with high flux and velocity in response to a high traffic demand, and thus the ramp meter is controlled to release more vehicles to the freeway. We consider that an adversary injects FDI attack at 300s such that the local controller switches to mode 2 (light traffic setting) and fewer vehicles are released from the ramp. Hence, the traffic flux decreases with the lower influx from the ramp, as shown in the middle plot of Fig~\ref{fig:case2rqv}. However, the velocity stays higher under the high traffic demand, leading to a lower traffic density $\rho$ closer to the inlet. The top and bottom plots of Fig~\ref{fig:case2rqv} illustrate these phenomena. Thus, the adversary forces the traffic not to settle to \(( q^3_\star , v^3_\star )\) and \textit{reduces the traffic throughput }for the freeway segment.
The top plot of Fig~\ref{fig:case2Res} shows the required (high) influx of vehicles from the ramp under nominal control scenarios $U_{{\sigma}} = U_3$. In contrast, after FDI attack, the compromised control input  $U_{\tilde{\sigma}} = U_2$ causes the ramp to stop releasing vehicles altogether. Under this FDI attack, the residual $r$ crosses the threshold within 50s of attack injection and the proposed scheme reliably detects the attack $\delta$, as shown in the last plot of Fig~\ref{fig:case2Res}.

\section{Conclusion} \label{con}
In this work, we adopt a model-based residual generation approach to detect cyberattacks on controller switching commands in ramp metering for an uncertain, multimodal macroscopic traffic operation governed by ARZ model. 
We propose a bank of detectors corresponding to each admissible traffic mode to compensate for the unknown traffic operational mode 
of freeway road networks under changing driving patterns. 
In addition, we utilize backstepping methodology along with Lyapunov functions to derive analytical design conditions for each of these detectors to guarantee stability, robustness to uncertainty, and sensitivity to switching attacks. We evaluate the performance of the proposed scheme against two prevalent cyberattacks: DoS and FDI, while considering realistic traffic parameters and uncertainties. The simulation results exhibit strong potential for real-time detection of cyberattacks in multi-mode traffic equipped with a hierarchical controller infrastructure.

\section*{Appendix}
\vspace{-2mm}
\textbf{Proof of Theorem \ref{DSS_th}:} \vspace{-2mm}
\subsection*{Lyapunov function and its derivatives} Let us consider a family of Lyapunov candidate functions for the detector bank such that the Lyapunov function for the $j^{th}$ detector is given as $E^j = E_1^j + E_2^j + E_3^j$,
where 
\begin{align}
     & E_1^j = \!\! \int_0^L \!\!\!\!\! e^{-x}\frac{{\phi^j}^2(x,t)}{2}dx,\,\, E_2^j = \!\! \int_0^L \!\!\!\!\! e^{x}\frac{{\psi^j}^2(x,t)}{2}dx, \label{lyap1}\\
    &  E_3^j = \frac{e^{-L} \xi}{2}{\phi^j}^2(L,t), \, \text{where}\,\,\xi = e^L-1>0. \label{lyap23}
\end{align} 
We first differentiate $E_1^j$ with respect to time:
\begin{align}
    \dot{E}^j_1  = \int_0^L e^{-x} \phi^j(x,t) \phi^j_t (x,t) dx. \label{e1Dot}
\end{align} We substitute $\phi^j_t (x,t)$ in \eqref{e1Dot} using \eqref{phi_tx} such that 
\begin{align}
    \dot{E}^j_1  =  I_1 + I_2 + I_3, \label{e1d_i123}
\end{align}  with 
 $\mathbb{f}^j \coloneqq   e^{-x} \phi^j(x,t)$, $I_1 =-  v^\alpha_\star \!\! \int_0^L \! \mathbb{f}^j \phi^j_x (x,t) dx$, \\$I_2  = \! \int_0^L \! \mathbb{f}^j a^j (x,t) dx$, and $I_3 = \! \int_0^L \! \mathbb{f}^j \theta_\phi^j (x,t) dx$. \\
Next, we use integration by parts twice on $I_1$ to obtain
\begin{align}
    I_1 &=   v^\alpha_\star \left[ 1 - e^{-L} \right] \phi^j(L,t)\phi^j_x (L,t) - \frac{ v^\alpha_\star e^{-L}}{2} {\phi^j}^2(L,t) \nonumber \\
    &\quad + \frac{ v^\alpha_\star}{2} {\phi^j}^2 (0,t) -   v^\alpha_\star E_1^j.
\end{align}
Using the boundary condition \eqref{BC_phi_zero} and Young's inequality, we can rewrite $I_1$ as
\begin{align}
    I_1 &=   v^\alpha_\star \left[ 1 - e^{-L} \right] \phi^j(L,t)\phi^j_x (L,t) - \frac{ v^\alpha_\star e^{-L}}{2} {\phi^j}^2(L,t) \nonumber \\
    &\quad + {\mu}_{9}  v^\alpha_\star\, h{^\alpha}^2\,{\psi^j}^{2} (0,t) + {\mu}_{10}  v^\alpha_\star\, \theta_h^2\,\hat{V}{^j}^2 (0,t) \nonumber\\
    &\quad + {\mu}_{11}  v^\alpha_\star\, \theta_q^2  + {\mu}_{12}  v^\alpha_\star\, {\delta}^2 -   v^\alpha_\star E_1^j.
\end{align}
Next, we simplify integral $I_2$ in four subsequent steps. 

\noindent
\textbf{Step I:} First, we apply Cauchy-Schwarz inequality and  Young's inequality on $I_2$ for  $\mu_{17} \in \mathbb{R}^+$, to obtain
\begin{align}
   & I_2 \leqslant \frac{1}{2 \mu_{17}}  E_1^j+ \frac{\mu_{17}}{2} \int_0^L {a^j}^2(x,t) dx. \label{i2}
\end{align}

\noindent
\textbf{Step II: } Using \eqref{a2dq} and Young's inequality, we can write
 \begin{align}
 \int_0^L\!\!\!\!\!\!  {a^j}^2(x,t) dx  & \leqslant  {\mu}_{19} \!\! \int_0^L \!\!\!\!\!\! {\eta_q}^2(x,t) dx + {\mu}_{19} \!\! \int_0^L \!\!\!\!\!\! \left[\mathcal{I}_{M^j} ( \eta_q)\right]^2 dx. \nonumber
\end{align}
Here, ${\mu}_{19} = \max(1+\mu_{18}, 1+ \mu_{18}^{-1})$.

\noindent
\textbf{Step III:}
Then, applying the Cauchy-Schwarz inequality on the squared inequality, we obtain
\begin{align}
    \int_0^L \!\!\!\!\!\! \left[\mathcal{I}_{M^j} ( \eta_q)\right]^2 dx \leqslant  \int_0^L \!\!\!\!\mathcal{I}_{{M^j}^2}\left(\eta_q^2 \right)  (L-x) dx.
\end{align}

\noindent
\textbf{Step IV:}
Moreover, from the kernel definitions \eqref{R_def} and \eqref{MN_def}, we can write { ${M^j}^2(x,z) \leqslant \max\limits_{x\leqslant z\leqslant L} {M^j}^2(x,z) \leqslant \left(\frac{h^\alpha}{\tau \gamma  p^\alpha_\star}\right)^2$}. Hence, for $2 \mu_7 = \mu_{17} \mu_{19} $ and $ \overline{M}^j = 1+ \frac{(L\,{h^\alpha})^2}{(\tau \gamma  p^\alpha_\star)^2} $, \eqref{i2} simplifies to
\begin{align}
    I_2 \leqslant \frac{1}{2\mu_{17}}  E_1^j+  \mu_7\, \overline{M}^j \lVert \eta_q \rVert^2.
\end{align}
\noindent
Furthermore, using \eqref{thetaPhiker} and following similar steps I to IV for $I_3$, we obtain
\begin{align}
    I_3 \leqslant \frac{1}{2\mu_{20}}  E_1^j+  {\mu}_{8} \overline{M}^j \left[ \theta_v^2 \lVert \hat{W}_x^j \rVert^2+ \theta_c^2\lVert k_1^j  \rVert^2\hat{W}{^j}^2(L,t) \right]. \nonumber
\end{align}
Here,  $\mu_{25}, \mu_{10} \in \mathbb{R}^+$. Now, substituting $I_1$, $I_2$, and $I_3$ in \eqref{e1d_i123}, we obtain:
\begin{align}
    \dot{E}^j_1 & \leqslant  -  \left[ v^\alpha_\star - \frac{1}{2\mu_{17}} - \frac{1}{2\mu_{20}}\right] E_1^j  - \frac{ v^\alpha_\star} {\xi} E_3^j + {\mu}_{12}  v^\alpha_\star\, {\delta}^2  \nonumber \\
    & + {\mu}_{9}  v^\alpha_\star\, h{^\alpha}^2\,{\psi^j}^2 (0,t) \!+ {\mu}_{10}  v^\alpha_\star\, \theta_h^2\,\hat{V}{^j}^2 (0,t) \!+ {\mu}_{11}  v^\alpha_\star\, \theta_q^2   \nonumber \\
    & +  v^\alpha_\star \left[ 1 - e^{-L} \right] \phi^j(L,t)\phi^j_x (L,t) +  \mu_7\, \overline{M}^j \lVert \eta_q \rVert^2 \nonumber \\
    &  + {\mu}_{8} \overline{M}^j \left[ \theta_v^2 \lVert \hat{W}_x^j \rVert^2+ \theta_c^2\lVert k_1^j  \rVert^2\hat{W}{^j}^2(L,t) \right]. \label{e1Dotf}
\end{align}
Next, we differentiate $E_2^j$ \eqref{lyap23} with respect to time and use \eqref{psi_tx} to obtain 
\begin{align}
    \dot{E}^j_2  =  I_4 + I_5 + I_6, \,\,\, \text{where} \label{e2_i456}
\end{align} 
\begin{align}
     I_4 =  v^\alpha_\star h^\alpha \!\!\! \int_0^L \!\!\!\! \mathbb{g}\psi^j_x dx, \,\,\, I_5  = \!\!\! \int_0^L \!\!\!\! \mathbb{g} b^j dx, \,\,\, I_6 = \!\!\!\int_0^L \!\!\!\! \mathbb{g} \theta_\psi^j dx, 
\end{align}
and  $\mathbb{g} = e^{x} \psi^j(x,t)$. Next, using integration by parts and the boundary condition \eqref{BC_psi_l}, $I_4$ becomes
\begin{align}
    I_4 & \leqslant \frac{\Upsilon_{19} }{2}   \phi{^j}^2 (L,t) +  \Upsilon_{19} \theta_k^2\,\hat{W}{^j}^2 (L,t) \nonumber\\
    & \qquad - \frac{ v^\alpha_\star  h^\alpha}{2} {\psi^j}^2 (0,t)-   v^\alpha_\star h^\alpha E_2^j.
\end{align}
Here, $\Upsilon_{19} = {\mu}_3  v^\alpha_\star  h^\alpha e^{L} {\lambda^j}^2$. We apply steps I to IV for $I_5$ and $I_6$ and substitute  $I_4$, $I_5$, and $I_6$ in \eqref{e2_i456} to get:
\begin{align}
    \dot{E}_2^j & \leqslant - \left[    v^\alpha_\star h^\alpha - \frac{1}{2 \mu_{21}} - \frac{1}{2 \mu_{22}} \right] E_2^j  + \frac{ \Upsilon_{19} e^L}{\xi}   E_3^j \nonumber   \\
    & - \frac{ v^\alpha_\star  h^\alpha}{2} {\psi^j}^2 (0,t) + {\mu}_5 \lVert \eta_v \rVert^2  + {{\mu}_5} \overline{N}^j \lVert \eta_q \rVert^2 \nonumber   \\
    &+ {\mu}_{13} \theta_{h,v}^2 \lVert \hat{V}^j_x  \rVert^2 + {\mu}_{13} \theta_{\beta}^2 \lVert \hat{W}^j  \rVert^2 + {\mu}_{13} \overline{N}^j \theta_v^2  \lVert \hat{W}^j_x  \rVert^2 \nonumber \\
    & +   \!  \left[{\mu}_{13} \theta_c^2 [ \lVert k_2^j  \rVert^2 \!+ \! \overline{N}^j  \lVert k_1^j  \rVert^2 ] \! + \!\!\Upsilon_{19}  \theta_k^2 \right]\!\! \hat{W}{^j}^2\!(L,t). \label{e2Dotf}
\end{align}
Here, 
$\overline{N}^j  = \frac{L^2}{\left(\tau \gamma  p^\alpha_\star\right)^2}$ is obtained from the kernel definitions \eqref{S_def} and \eqref{MN_def} in Step IV and for ${N^j}^2(x,z) \leqslant \max\limits_{x\leqslant z\leqslant L} {N^j}^2(x,z) \leqslant \frac{1}{\left(\tau \gamma  p^\alpha_\star\right)^2}$.
Now,  we differentiate $E_3^j$ \eqref{lyap23} with respect to time and use \eqref{phi_tx} to obtain 
\begin{align}
    \dot{E}_3^j &  \leqslant  - v^\alpha_\star e^{-L} \xi \phi^j (L,t) \phi^j_x (L,t) + \frac{{\mu}_2 e^{-L}}{\xi} E_3^j \nonumber \\
    & \quad +\frac{ \mu_{6} \xi{^j}^2}{2} a{^j}^2 (L,t) +\frac{ \mu_{23} \xi{^j}^2}{2} \theta_\phi^2 (L,t). 
\end{align}
Next, using \eqref{a2dq}-\eqref{thetaPhiker},  applying Young's inequality, and for $2 \mu_{14} = \mu_{23}\max(1+\mu_{24}, 1+ \mu_{24}^{-1}k{_1^j}^2(L))$, $\dot{E}_3^j$ yields
\begin{align}
    &  \dot{E}_3^j   \leqslant  - v^\alpha_\star e^{-L} \xi \phi^j (L,t) \phi^j_x (L,t)  +\frac{ \mu_{6} \xi{^j}^2}{2} \eta_q^2 (L,t)\nonumber \\
    & + \frac{{\mu}_2 e^{-L}}{\xi} E_3^j \! + {\mu}_{14}\! \left[\theta_v^2 \hat{W}{^j_x}^2 (L,t)
     +  \theta_c^2  \hat{W}{^j}^2 (L,t) \right]. \label{e3Dotf}
\end{align}
Now, using \eqref{e1Dotf}, \eqref{e2Dotf}, and \eqref{e3Dotf}, we obtain $\dot{E}^j$ as:
\begin{align}
    \dot{E}^j & \leqslant -  \Upsilon_{20}^\alpha E_1^j   - \Upsilon_{21}^\alpha E_2^j  - \Upsilon_{8}^{\alpha,j}  E_3^j + {\mu}_{12}  v^\alpha_\star\, {\delta}^2 \nonumber   \\ 
    &  + \Upsilon_{13}^j \lVert \eta_q \rVert^2  + {\mu}_5 \lVert \eta_v \rVert^2 +\frac{ \mu_{6} \xi{^j}^2}{2} \eta_q^2 (L,t)  \nonumber \\
    & - \Upsilon_{22}^\alpha {\psi^j}^2 (0,t)+ \Upsilon_{23}^\alpha \phi^j(L,t)\phi^j_x (L,t)  +\Upsilon_\theta \Theta^j. \label{eDot_short}
\end{align}
The parameters \eqref{eDot_short}  are defined in Table~\ref{tab:thm_var} and Table~\ref{tab:eDot_var}.
\renewcommand{\arraystretch}{1.4}
\begin{table}[h!]
    \centering
    \caption{List of parameters in \eqref{eDot_short}.}
    \begin{tabular}{|l| l|}
    \hline
        $\Upsilon_{20}^\alpha  =  v^\alpha_\star - \frac{1}{2\mu_{18}} - \frac{1}{2\mu_{20}}$ & $\Upsilon_{21}^\alpha =  v^\alpha_\star h^\alpha - \frac{1}{2 \mu_{21}} - \frac{1}{2\mu_{22}}$ \\
        \hline
        $\Upsilon_{22}^\alpha = \frac{ v^\alpha_\star  h^\alpha}{2} \left[ 1- 2{\mu}_{9}  h{^\alpha}\right]$  &  $\Upsilon_{23}^\alpha =  v^\alpha_\star \left[ 1 - e^{-L} - e^{-L} \xi \right]$ \\
        \hline
          \multicolumn{2}{|l|}{$\Theta^j(t) \coloneqq$  $\max \{ \theta_h^2\,\hat{V}{^j}^2 (0,t), \theta_{h,v}^2 \lVert \hat{V}^j_x  \rVert^2, \theta_{\beta}^2 \lVert \hat{W}^j  \rVert^2,$} \\
         \multicolumn{2}{|c|}{$\quad \theta_v^2 \lVert \hat{W}_x^j \rVert^2,\theta_q^2, \theta_v^2 \hat{W}{^j_x}^2 (L,t),\theta_c^2  \hat{W}{^j}^2 (L,t)\}\geqslant 0 $ }\\
        \hline
    \end{tabular}
        \label{tab:eDot_var}
\end{table}\\
\noindent
Since $\xi = e^L -1$ by choice, $\Upsilon_{23}^\alpha = 0, \forall \alpha$. Furthermore,  for $2{\mu}_{9}  h{^\alpha} \leqslant 1$, we can obtain
\begin{align}
    \dot{E}^j & \leqslant -  \Upsilon_{20}^\alpha E_1^j   -\! \Upsilon_{21}^\alpha E_2^j  -\! \Upsilon_{8}^{\alpha,j}  E_3^j + {\mu}_{12}  v^\alpha_\star\, {\delta}^2 \! +\! \Upsilon_{13}^j \lVert \eta_q \rVert^2 \nonumber   \\ 
    &  \quad + {\mu}_5 \lVert \eta_v \rVert^2 +\frac{ \mu_{6} \xi{^j}^2}{2} \eta_q^2 (L,t)  +\Upsilon_\theta \Theta^j. \label{eDotFinal}
\end{align}
Next,  utilizing $E^j$ \eqref{lyap1}-\eqref{lyap23} and the derivatives \eqref{eDotFinal}, we will prove the theorem in four parts to derive the conditions for AURS, ES, robustness, and sensitivity.

\subsection*{Condition for AURS}
During compromised traffic operation  in the presence of uncertainties, i.e., $\delta,\eta_q,\eta_v \neq 0$,  $E^j$ from \eqref{eDotFinal} becomes:
\begin{align}
    \dot{E}^j \! & \leqslant \!\! -  \Upsilon_{24}^j E^j  \! +\! \Upsilon_{25}^j \!\left[ {\delta}^2 \!\! +\! \overline{\eta}^2\right ]\! +\! \Upsilon_\theta \Theta^j . \label{eDotDSS}
\end{align}
Here, $\Upsilon_{25}^j = \max({\mu}_{12}  v^\alpha_\star, \Upsilon_{13}^j,{\mu}_5,\frac{ \mu_{6} \xi{^j}^2}{2})$ and $\Upsilon_{24}^j = \min(\Upsilon_{20}^\alpha, \Upsilon_{21}^\alpha, \Upsilon_{8}^{\alpha,j}), \forall \alpha$. We define 
$\overline{\eta}^2 = \! \lVert \eta_q \rVert^2\!\! +\!\!\lVert \eta_v \rVert^2\! \!+ \!\eta_q^2 (L)$. Furthermore, using the ${\lambda}^j$ definition from \eqref{lambda}, we can deduce that the condition on $k_3^j$ as \eqref{dssThc} ensures $\Upsilon_{8}^{\alpha,j} \geqslant 0$ for ${\mu}_1 \in \mathbb{R}^+$. Thus, \eqref{dssThc} ensures the convergence of the Lyapunov functions $E^j$ in \eqref{eDotDSS}. Next, using Gronwall's inequality we obtain 
\begin{align}
    E^j (t) & \leqslant E^j (0) e^{-\Upsilon_{24}^j\,t} \!+\!\frac{\Upsilon_{25}^j}{\Upsilon_{24}^j} \sup\limits_{t\geqslant0}  \!\left[ {\delta}^2 \! +\! \overline{\eta}^2\right ]\!+ \!\frac{\Upsilon_\theta}{\Upsilon_{24}^j} \sup\limits_{t\geqslant0} \Theta^j\!. \nonumber
\end{align}
Furthermore, adding $\frac{e^{-L}\xi}{2} \theta_c^2 \hat{W}{^j}^2(L,t)$ on both side, we ca rewrite the above inequality as 
\begin{align}
    E^j (t) + \frac{e^{-L}\xi}{2} \theta_c^2 \hat{W}{^j}^2(L,t)& \leqslant E^j (0) e^{-\Upsilon_{24}^j\,t} + \overline{\Upsilon}_\theta \sup\limits_{t\geqslant0} \Theta^j \nonumber \\
    & \,\quad  +\frac{\Upsilon_{25}^j}{\Upsilon_{24}^j} \sup\limits_{t\geqslant0}  \!\left[ {\delta}^2 \! +\! \overline{\eta}^2\right ]. \label{etE0hatW}
\end{align}
Here $\overline{\Upsilon}_\theta = {\max\left(\Upsilon_\theta, \,\, \Upsilon_{15}+\frac{e^{-L}\xi}{2} \right)}\Upsilon{_{20}^j}^{-1}$. Now, from Assumption \ref{IC_assm}, we have $E^j(0) \! \leqslant\!\! \Upsilon_{26}^j {\phi}{^j_{0}}^2(L)$ for $\Upsilon_{26}^j= (\frac{\overline{\phi}{^j_{0}}^2(x) + \overline{\psi}{^j_{0}}^2 (x)e^L}{\underline{\phi}{^j_{0}}^2(L)} + \frac{\xi}{2e^{L}} )$. Moreover, from Assumption~\ref{IC_assm}, we can assume that there exits a constant $\Upsilon_0 \in \mathbb{R}^+$, such that $\Upsilon_0 r^2 (0) \geqslant \overline{\Upsilon}_0  {\zeta^j}^2 (0) \geqslant {\phi^j}^2(L,0) $. Next, from \eqref{Err_phi_psi}, we can write $ \frac{e^{-L}\xi}{2 {\mu}_{32}}  {\zeta^j}^2 (t)  \leqslant    E^j (t) +  \frac{e^{-L}\xi}{2} \theta_c^2 \hat{W}{^j}^2(L,t)$, where ${\mu}_{32} > 0$ from Young's inequality. This implies
\begin{align}
    \zeta{^j}^2 (t) & \leqslant \overline{\Upsilon}_2 \zeta{^j}^2 (0)  e^{-\Upsilon_3\,t} + {\Upsilon_4} \sup\limits_{t\geqslant0}  \!\left[ {\delta}^2 \! +\! \overline{\eta}^2\right ]\nonumber  \\
    & \qquad \qquad +\overline{\Upsilon}_\theta \sup\limits_{t\geqslant0} \Theta^j. \label{RjtR0} 
\end{align}
where $\overline{\Upsilon}_2 = \max\limits_j \left(\Upsilon_{26}^j\frac{2 {\mu}_{32} \Upsilon_0 e^L}{\xi} \right) $, $\Upsilon_3 = \min\limits_j \Upsilon_{24}^j$, and $\Upsilon_4 = \max\limits_j\frac{\Upsilon_{25}^j}{\Upsilon_{24}^j} $. 
Now, from the definition of $\Theta^j$ (Table~\ref{tab:eDot_var}),
\begin{align}
    \min\limits_j\Theta^j  = 0, \label{theta 0}
\end{align}
 since one of the detector parameters must always match the traffic dynamics parameters.
  Thus, taking minimum on both side of \eqref{RjtR0} and for $\Upsilon_2 = \overline{\Upsilon}_2 \Upsilon_0$, we obtain \eqref{DSS_criteria}.
Thus, we prove that \eqref{dssThc} provides a sufficient condition to satisfy \eqref{DSS_criteria} and ensure stability of the detection scheme under all circumstances. 

\subsection*{Condition for ES} During nominal traffic operation without uncertainty, i.e., $\delta, \eta_q, \eta_v \neq 0$, we obtain $\Upsilon_1 \geqslant \Upsilon_3$ in the absence of the tuning parameters from $\delta, \eta_q, \eta_v$ in $\Upsilon_{20}^\alpha, \Upsilon_{21}^\alpha, \Upsilon_{8}^{\alpha,j}$ (Table~\ref{tab:thm_var},~\ref{tab:eDot_var}). The proof then follows directly from the proof of AURS and \eqref{dssThc} provides a sufficient condition for the exponential convergence of the residual $r(t)$ \eqref{ES_criteria}. 

\subsection*{Condition for robustness}
Let us first prove the robustness criteria for each detector such that for the $j^{th}$ detector we have the criteria as \begin{align}
    \! \int_0^\infty \!\!\!\!\!\! \zeta{^j}^2 dt \leqslant \! \Upsilon_{5} \!\! \int_0^\infty \!\!\!\!\!\!  \left[ \overline{\eta}^2 \!+ \!\Theta^j \right ] dt +\! \epsilon. \label{robust def zeta}
\end{align}
Now, from the definition  of $\zeta^j$ \eqref{Err_phi_psi} and $\Theta^j$ (Table \ref{tab:eDot_var}), we have $\zeta{^j}^2 \leqslant \mu_4 \phi{^j}^2 (L,t)+ \mu_4 \theta_c^2 \hat{W}{^j}^2(L,t) \leqslant  \mu_4 \phi{^j}^2 (L,t) + \mu_4 \Theta^j$. This inequality implies that if we prove 
\begin{align}
    \int_0^\infty \!\!\!\!\!\! \mu_4 \phi{^j}^2 (L,t) dt & \leqslant \!\! \int_0^\infty \!\!\!\!\!\!  [ \Upsilon_{5} \overline{\eta}^2 +   [\Upsilon_{5} -  \mu_4 \Theta^j] ] dt \!+\! \epsilon, \label{robust def phi}
\end{align}
it in turn proves \eqref{robust def zeta}. Next, let us define a vector $\Gamma_1 = \begin{bmatrix}
    \phi^j (L,t) & \lVert \eta_v \rVert & \lVert \eta_q \rVert & \eta_q (L,t) & \Theta
\end{bmatrix}$. Then,  \eqref{robust def phi} becomes
\begin{align}
    \int_0^\infty  \Gamma_1 \mathcal{D}_1 \Gamma_1^T dt - \epsilon \leqslant 0, \label{lambda rob}
\end{align} using $\Gamma_1$, where $ \mathcal{D}_1 = \text{diag} \{\mu_4,  - \Upsilon_{5}, $ $- \Upsilon_{5}, -\Upsilon_{5},  \mu_4 - \Upsilon_{5}\}$ is a $5\times5$ matrix. Furthermore, during nominal traffic operation and in the presence of uncertainties ($\delta =0, \eta_q, \eta_v \neq 0$), the derivative of Lyapunov functions $E^j$  from \eqref{eDotFinal} becomes $ \dot{E}^j \leqslant   \Gamma_1 \mathcal{B}_1 \Gamma_1^T$, where $\mathcal{B}_1 = \Lambda_1^j -\mathcal{D}_1$ and $\Lambda_1^j$ is defined in \eqref{robLam}. Then, integrating this inequality, we obtain
\begin{align}
    E^j (\infty) - E^j (0) - \int_0^\infty  \Gamma_1 \mathcal{B}_1 \Gamma_1^T dt \leqslant 0. \label{einfRob}
\end{align}
Since the left side of \eqref{einfRob} is non-positive, so if we prove
\begin{align}
    \int_0^\infty \! \!\!\!\!\! \Gamma_1 \mathcal{D}_1 \Gamma_1^T dt - \epsilon \leqslant E^j (\infty)\! - \!E^j (0)\! - \!\!\!\int_0^\infty \! \!\!\!\!\! \Gamma_1 \mathcal{B}_1 \Gamma_1^T dt, \label{d1b1E}
\end{align}
it will be sufficient to prove \eqref{lambda rob}. Next, we rearrange \eqref{d1b1E}
\begin{align}
    \int_0^\infty  \!\!\!\Gamma_1 \left[ \mathcal{D}_1 +\mathcal{B}_1  \right] \,\Gamma_1^T dt \leqslant E^j (\infty) - E^j (0) + \epsilon. \label{d1minusb1}
\end{align}
Note that $E^j (\infty) \geqslant 0$ by definition. Moreover, from Assumption \ref{IC_assm} we can deduce that $E^j(0) - \epsilon \leqslant 0$ with our choice of $ \epsilon \geqslant \max\limits_j \Upsilon_{20}^j$.  Hence, we get
\begin{align}
    E^j (\infty) - E^j (0) + \epsilon \geqslant 0. \label{eInfe0}
\end{align} 
Therefore, if we can guarantee $\mathcal{D}_1 +\mathcal{B}_1 = \Lambda_1^j \leqslant 0$, then it will in turn prove \eqref{robust def zeta} using \eqref{d1minusb1}, \eqref{d1b1E}, and \eqref{robust def phi}.  Furthermore, taking the minimum on both sides of \eqref{robust def zeta} and using \eqref{theta 0}, we obtain our robustness criteria \eqref{robust def r}. This implies \eqref{rbostThc} is a sufficient condition to satisfy \eqref{robust def r}.

\subsection*{Condition for sensitivity}\vspace{-2mm}
Next, we first prove the sensitivity criteria for each detector and for  the $j^{th}$ detector we have the criteria as 
\begin{align}
    \int_0^\infty \zeta{^j}^2 dt \geqslant \Upsilon_{6} \int_0^\infty \left[ \delta^2 + \Theta^j \right] dt - \epsilon. \label{sen def zeta}
\end{align}
Next, we add $ \int_0^\infty\theta_c^2 \hat{W}{^j}^2(L,t) dt$ on both sides of \eqref{sen def zeta}
\begin{align}
    &\int_0^\infty  \!\!\!\!\! [\zeta{^j}^2 +\theta_c^2 \hat{W}{^j}^2(L,t) ] dt \geqslant  \nonumber \\
    & \quad \qquad \int_0^\infty \!\!\!\!\! [ \Upsilon_{6} \delta^2 + \Upsilon_{6} {\Theta}^2 +\theta_c^2 \hat{W}{^j}^2(L,t)] dt - \epsilon. \label{sen def zetaW}
\end{align} 
Next, from the definition  of $\zeta^j$ \eqref{Err_phi_psi} we have ${\mu}_{20}  \zeta{^j}^2 + {\mu}_{20}  \theta_c^2 \hat{W}{^j}^2(L,t) \geqslant c{^\alpha}^2 \phi{^j}^2 (L,t)$. Furthermore, since $ {\Theta}^2 >\theta_c^2 \hat{W}{^j}^2(L,t)$, 
to prove \eqref{sen def zeta} using \eqref{sen def zetaW}, it is sufficient to prove
\begin{align}
    \int_0^\infty \!\!\!\frac{c{^\alpha}^2}{{\mu}_{20}}  \phi{^j}^2 (L,t)dt \geqslant  \int_0^\infty \!\!\!\!\!\! [ \Upsilon_{6} \delta^2 \!+ \!({\Upsilon}_{6}\!+\!1) {\Theta}^2 ] dt - \epsilon. \label{sen def phiSQ}
\end{align}
Next,  we add $\int_0^\infty [ {\mu_{15}}  \delta^2 + {\mu_{16}} {\Theta}^2 ] dt$ on both sides of \eqref{sen def phiSQ}, where $\mu_{15},\mu_{16}\in \mathbb{R}^+$ is such that $\mu_{15}^{-1}+ \mu_{16}^{-1} = \frac{ c{^\alpha}^2}{{\mu}_{20}} -\Upsilon_{16} $. Then, using Young's inequality on the left side of the inequality, we have $ \frac{c{^\alpha}^2}{{\mu}_{20}}  \phi{^j}^2 (L,t) + {\mu_{15}}  \delta^2 + {\mu_{16}}  {\Theta}^2 \geqslant \Upsilon_{16} \phi{^j}^2 (L,t) + 2\phi{^j} (L,t) \delta + 2 \phi{^j} (L,t) {\Theta}$.
Therefore, to prove \eqref{sen def phiSQ}, we can prove 
\begin{align}
   & \int_0^\infty \!\! \left[ \Upsilon_{16} \phi{^j}^2 (L,t) +  2\phi{^j} (L,t) \delta + 2\phi{^j} (L,t) {\Theta} \right] dt \geqslant \nonumber \\
    &  \quad  \int_0^\infty \!\!\!\!\! \left[ (\Upsilon_{6} +{ \mu_{15}}) \delta^2 + (\Upsilon_{6} +{\mu_{16}} +1) {\Theta}^2 \right] dt - \epsilon. \label{sen def phi coupled}
\end{align}
Now, let us define a vector $\Gamma_2 = \begin{bmatrix}
    \phi^j (L,t) & \delta  & {\Theta}
\end{bmatrix}$. Then, we rewrite \eqref{sen def phi coupled} using $\Gamma_2$ as 
\begin{align}
    \int_0^\infty  \Gamma_2 \mathcal{D}_2 \Gamma_2^T dt + \epsilon \geqslant 0, \label{lambda sen def}
\end{align} where $ \mathcal{D}_2 = \begin{bmatrix}
    \Upsilon_{16}  & 1 & 1 \\
    1 & \quad -\Upsilon_{6} - { \mu_{15}} & 0 \\
    1 & 0 & \,\, -\Upsilon_{6}-{\mu_{16}} -1
\end{bmatrix}$.
On the other hand, under cyberattack without uncertainty ($\delta \neq0, \eta_q,\eta_v =0$) and since $\overline{\Upsilon}_8^{\alpha,j}\geqslant -\overline{\Upsilon}_8^{\alpha,j}$, we have the derivative of Lyapunov functions $E^j$ from \eqref{eDotFinal} as $ \dot{E}^j \leqslant   \Gamma_2 \mathcal{B}_2 \Gamma_2^T$, where $\mathcal{B}_2 = \Lambda_2^j  + \mathcal{D}_2$ and $\Lambda_2^j$ is defined in \eqref{senLam}. Then, integrating this inequality, we obtain
\begin{align}
     \int_0^\infty  \Gamma_2 \mathcal{B}_2 \Gamma_2^T dt - E^j (\infty) + E^j (0) \geqslant 0. \label{einfSen}
\end{align}
\eqref{einfSen} implies that to prove \eqref{lambda sen def}, it is sufficient to prove
     $\int_0^\infty   \Gamma_2 \mathcal{D}_2 \Gamma_2^T dt + \epsilon \geqslant \int_0^\infty   \Gamma_2 \mathcal{B}_2 \Gamma_2^T dt - E^j (\infty) + E^j (0).$ 
We next we rearrange this inequality to obtain 
\begin{align}
    E^j (\infty) - E^j (0) + \epsilon \geqslant \int_0^\infty  \! \Gamma_2 \left[ \mathcal{B}_2 - \mathcal{D}_2  \right] \,\Gamma_2^T dt. \label{senBD}
\end{align}
    Moreover,  \eqref{eInfe0} implies that  $\mathcal{B}_2 - \mathcal{D}_2 = \Lambda_2^j \leqslant 0$ is sufficient condition to prove \eqref{senBD} which in turn proves \eqref{sen def zeta}. Furthermore, taking the minimum on both sides of \eqref{sen def zeta} and using \eqref{theta 0}, we obtain the sensitivity criteria \eqref{sen def r}. Thus, we prove that \eqref{sensThc} provides a sufficient condition to satisfy \eqref{sen def r} and ensures attack-sensitive residual generation for the proposed scheme. \hfill $\blacksquare$

\bibstyle{arxiv}
\bibliography{ref.bib, Ref_Troy}

\begin{thebibliography}{10}

\bibitem{cassandras2016smart}
Christos~G Cassandras.
\newblock \capitalisewords{Smart cities as cyber-physical social systems}.
\newblock {\em Engineering}, 2(2):156--158, 2016.

\bibitem{perrine2019implications}
Kenneth~A Perrine, Michael~W Levin, Cesar~N Yahia, Melissa Duell, and Stephen~D Boyles.
\newblock \capitalisewords{Implications of traffic signal cybersecurity on potential deliberate traffic disruptions}.
\newblock {\em Transportation research part A: policy and practice}, 120:58--70, 2019.

\bibitem{roy2020secure}
Tanushree Roy and Satadru Dey.
\newblock Secure traffic networks in smart cities: Analysis and design of cyber-attack detection algorithms.
\newblock In {\em 2020 American Control Conference (ACC)}, pages 4102--4107. IEEE, 2020.

\bibitem{roy2021socio}
Tanushree Roy, Amara Tariq, and Satadru Dey.
\newblock A socio-technical approach for resilient connected transportation systems in smart cities.
\newblock {\em IEEE Transactions on Intelligent Transportation Systems}, 23(6):5019--5028, 2021.

\bibitem{ghosh2023security}
Sanchita Ghosh and Tanushree Roy.
\newblock Security of cyber-physical systems under compromised switching.
\newblock In {\em 2023 IEEE Conference on Control Technology and Applications (CCTA)}, pages 1034--1039. IEEE, 2023.

\bibitem{fan2013comparative}
Shimao Fan, Michael Herty, and Benjamin Seibold.
\newblock \capitalisewords{Comparative model accuracy of a data-fitted generalized Aw-Rascle-Zhang model}.
\newblock {\em arXiv preprint arXiv:1310.8219}, 2013.

\bibitem{tampere2007extended}
Chris~MJ Tamp{\`e}re and LH~Immers.
\newblock \capitalisewords{An extended Kalman filter application for traffic state estimation using CTM with implicit mode switching and dynamic parameters}.
\newblock In {\em \capitalisewords{2007 IEEE intelligent transportation systems conference}}, pages 209--216. IEEE, 2007.

\bibitem{gu2022distributionally}
Chuanye Gu, Changzhi Wu, Yonghong Wu, and Benchawan Wiwatanapataphee.
\newblock \capitalisewords{Distributionally robust ramp metering under traffic demand uncertainty}.
\newblock {\em Transportmetrica B: Transport Dynamics}, 10(1):652--666, 2022.

\bibitem{zheng2023urban}
Liang Zheng, Ji~Bao, and Zhenyu Mei.
\newblock \capitalisewords{Urban traffic signal control robust optimization against risk-averse and worst-case cyberattacks}.
\newblock {\em Information Sciences}, 640:119067, 2023.

\bibitem{li2005distributed}
Wei Li and Christos~G Cassandras.
\newblock \capitalisewords{Distributed cooperative coverage control of sensor networks}.
\newblock In {\em \capitalisewords{Proceedings of the 44th IEEE Conference on Decision and Control}}, pages 2542--2547. IEEE, 2005.

\bibitem{tsitsokas2023two}
Dimitrios Tsitsokas, Anastasios Kouvelas, and Nikolas Geroliminis.
\newblock \capitalisewords{Two-layer adaptive signal control framework for large-scale dynamically-congested networks: Combining efficient max pressure with perimeter control}.
\newblock {\em Transportation Research Part C: Emerging Technologies}, 152:104128, 2023.

\bibitem{cazorla2022ramp}
Patricia Cazorla, Francisco Calder{\'o}n, and Elina Avila-Ord{\'o}{\~n}ez.
\newblock \capitalisewords{Ramp Metering Strategies: A Literature Review}.
\newblock {\em Revista Polit{\'e}cnica}, 50(1):15--26, 2022.

\bibitem{reilly2016creating}
Jack Reilly, S{\'e}bastien Martin, Mathias Payer, and Alexandre~M Bayen.
\newblock \capitalisewords{Creating complex congestion patterns via multi-objective optimal freeway traffic control with application to cyber-security}.
\newblock {\em Transportation Research Part B: Methodological}, 91:366--382, 2016.

\bibitem{ghena2014green}
Branden Ghena, William Beyer, Allen Hillaker, Jonathan Pevarnek, and J~Alex Halderman.
\newblock \capitalisewords{Green lights forever: Analyzing the security of traffic infrastructure}.
\newblock In {\em \capitalisewords{8th USENIX workshop on offensive technologies (WOOT 14)}}, 2014.

\bibitem{wang2025effects}
Liangwen Wang, Heng Ding, Xiaoyan Zheng, and Weihua Zhang.
\newblock \capitalisewords{Effects of Cyberattacks on Regional Traffic Networks in a Connected Vehicle Environment}.
\newblock {\em IEEE Transactions on Network Science and Engineering}, 2025.

\bibitem{dong2024evaluating}
Changyin Dong, Yujia Chen, Hao Wang, Leizhen Wang, Ye~Li, Daiheng Ni, De~Zhao, and Xuedong Hua.
\newblock \capitalisewords{Evaluating impact of remote-access cyber-attack on lane changes for connected automated vehicles}.
\newblock {\em Digital Communications and Networks}, 10(5):1480--1492, 2024.

\bibitem{ghosh2024cyberattack}
Sanchita Ghosh, Nutan Saha, and Tanushree Roy.
\newblock \capitalisewords{A cyberattack detection-isolation algorithm for CAV under changing driving environment}.
\newblock {\em IEEE Transactions on Intelligent Transportation Systems}, 2024.

\bibitem{ghosh2025assessment}
Sanchita Ghosh and Tanushree Roy.
\newblock \capitalisewords{Assessment of Cyberattack Detection-Isolation Algorithm for CAV Platoons Using SUMO}.
\newblock {\em arXiv preprint arXiv:2503.14628}, 2025.

\bibitem{roy2023redundancies}
Tanushree Roy, Sara Sattarzadeh, and Satadru Dey.
\newblock Cyber-attack detection in socio-technical transportation systems exploiting redundancies between physical and social data.
\newblock {\em IEEE Transactions on Systems, Man, and Cybernetics: Systems}, pages 1--12, 2023.

\bibitem{menelaou2024stealthy}
Charalambos Menelaou, Kangkang Zhang, Stelios Timotheou, Christos~G Panayiotou, and Thomas Parisini.
\newblock \capitalisewords{Stealthy Attack Detection of controlled ramp meters in freeway networks}.
\newblock In {\em \capitalisewords{2024 IEEE 63rd Conference on Decision and Control (CDC)}}, pages 1289--1294. IEEE, 2024.

\bibitem{canepa2013spoofing}
Edward~S Canepa and Christian~G Claudel.
\newblock \capitalisewords{Spoofing cyber attack detection in probe-based traffic monitoring systems using mixed integer linear programming}.
\newblock In {\em \capitalisewords{2013 International Conference on Computing, Networking and Communications (ICNC)}}, pages 327--333. IEEE, 2013.

\bibitem{kashyap2021detection}
Abhishek Kashyap, Animesh Chakravarthy, and Prathyush~P Menon.
\newblock \capitalisewords{Detection of cyber-attacks in automotive traffic using macroscopic models and Gaussian processes}.
\newblock {\em IEEE Control Systems Letters}, 6:1688--1693, 2021.

\bibitem{adelipour2025vulnerability}
Saeed Adelipour, Enayatollah Amiri~Darreh Razgahi, and Mohammad Haeri.
\newblock \capitalisewords{Vulnerability Mitigation of Urban Traffic Control Against Cyberattacks Using Secure Multi-Party Computation}.
\newblock {\em IEEE Transactions on Intelligent Transportation Systems}, 2025.

\bibitem{shi2021physics}
Rongye Shi, Zhaobin Mo, Kuang Huang, Xuan Di, and Qiang Du.
\newblock \capitalisewords{A physics-informed deep learning paradigm for traffic state and fundamental diagram estimation}.
\newblock {\em IEEE Transactions on Intelligent Transportation Systems}, 23(8):11688--11698, 2021.

\bibitem{aw2000resurrection}
A.A.T.M. Aw and Michel Rascle.
\newblock \capitalisewords{Resurrection of" second order" models of traffic flow}.
\newblock {\em SIAM journal on applied mathematics}, 60(3):916--938, 2000.

\bibitem{Hespanha2003}
J.~P. Hespanha, Daniel Liberzon, and A.~Stephen Morse.
\newblock Hysteresis-based switching algorithms for supervisory control of uncertain systems.
\newblock {\em Automatica}, 39:263--272, 2 2003.

\bibitem{song2025mode}
Xiaona Song, Zenglong Peng, Shuai Song, Xin Wang, and Xiaohui Zhang.
\newblock \capitalisewords{Mode-dependent spatiotemporal event-triggered control for switched nonlinear PDE systems with persistent dwell-time switching regulation}.
\newblock {\em Journal of Control and Decision}, pages 1--17, 2025.

\bibitem{yu2019traffic}
Huan Yu and Miroslav Krstic.
\newblock \capitalisewords{Traffic congestion control for Aw--Rascle--Zhang model}.
\newblock {\em Automatica}, 100:38--51, 2019.

\bibitem{shi2021constructing}
Xiaowei Shi and Xiaopeng Li.
\newblock \capitalisewords{Constructing a fundamental diagram for traffic flow with automated vehicles: Methodology and demonstration}.
\newblock {\em Transportation Research Part B: Methodological}, 150:279--292, 2021.

\bibitem{nigam2022empirical}
Archana Nigam, Manish Chaturvedi, and Sanjay Srivastava.
\newblock \capitalisewords{An empirical study on parameters affecting traffic stream variables under rainy conditions}.
\newblock In {\em \capitalisewords{2022 14th International Conference on COMmunication Systems \& NETworkS (COMSNETS)}}, pages 818--823. IEEE, 2022.

\bibitem{Ding}
Steven~X Ding.
\newblock {\em \capitalisewords{Model-based fault diagnosis techniques: design schemes, algorithms, and tools}}.
\newblock Springer Science \& Business Media, 2008.

\bibitem{Troy_actuator_anomaly}
Tanushree Roy and Satadru Dey.
\newblock Actuator anomaly detection in linear parabolic distributed parameter cyber-physical systems.
\newblock {\em IEEE Transactions on Control Systems Technology}, pages 1--12, 2023.

\end{thebibliography}

\end{document}